%% file: ex_article.tex
\documentclass[final,onefignum,onetabnum]{siamonline171218}

\input{ex_shared}

\ifpdf
\hypersetup{
  pdftitle={\TheTitle},
  pdfauthor={\TheAuthors}
}
\fi

\begin{document}

\maketitle

\begin{abstract}
To study the effect of boundaries on diffusion of new products, we introduce two novel analytic tools: The {\em indifference principle}, which enables us to explicitly compute the aggregate diffusion on various networks,
and the {\em dominance principle}, which enables us to rank the diffusion on different networks.
Using these principles, we prove our main result that on a finite line, one-sided diffusion (i.e., when each consumer can only be influenced by her left neighbor)
is strictly slower than two-sided diffusion (i.e., when each consumer can be influenced by her left and right neighbor). 
This is different from the periodic case of diffusion on a circle, where 
one-sided and two-sided diffusion are identical. We observe numerically similar results in higher dimensions.  
\end{abstract}



\section{Introduction}
Diffusion of new products is a fundamental problem in marketing~\cite{Mahajan-93}. More generally, 
diffusion in social networks has attracted the attention of researchers  
in physics, mathematics, biology, computer science, social sciences, economics, and management science,
as it concerns the spreading of ``items'' ranging from diseases and computer viruses to rumors, information, opinions, technologies and 
innovations~\cite{Albert-00,Anderson-92,Jackson-08,Pastor-Satorras-01,Rogers-03,Strang-98}. 
The first mathematical model of diffusion of new products
was proposed in~1969 by Bass~\cite{Bass-69}.  In the Bass model, an individual adopts a new product because 
of {\em external influences} by mass media, and {\em internal influences}
by individuals who have already adopted the product.
Bass wrote a single ODE for the number of adopters in the market as a function of time, and showed that its solution has an S-shape.
This classical paper inspired a
huge body of theoretical and empirical research~\cite{Hopp-04}.
Most of these studies also used ODEs to model the adoption level of the whole market. 
More recently, diffusion of new products has been studied 
using discrete, agent-based models (ABM)~\cite{OR-10,Garcia-05,Goldenberg-09,GLM-02,GLM-10}. 
This kinetic-theory approach has the advantage that it reveals the relation between the behavior of individual consumers 
and the aggregate market diffusion. In particular, discrete models can allow for a social network structure, whereby individuals are only influenced by their peers. Most studies that used a discrete Bass model with a spatial structure were numerical (e.g.,~\cite{Delre-07,Garber-04,GLM-02}). To the best of our knowledge, an analysis of the effect of the network
structure on the diffusion of new products was only done in~\cite{OR-10} for the discrete Bass model,
and in~\cite{Bass-SIR-analysis-17} for the discrete Bass-SIR model.
Thus, at present there is limited understanding of the effect of the social network structure  on the diffusion.

  In this paper we present the first study of boundary effects in the discrete Bass model.
Our motivation comes from 
products that diffuse predominantly through internal influences by geographical neighbors,
such as residential rooftop solar systems~\cite{Bollinger-12,Graziano-15}. For such products, it is reasonable to approximate the social network by a two-dimensional grid, in which each node represents a residential unit. Previous studies of the discrete Bass model 
on two-dimensional networks avoided boundary effects by 
imposing periodic boundary conditions~\cite{Bass-SIR-model-16,Bass-SIR-analysis-17,OR-10}. 
In that case, all nodes are interior and are influenced  from all sides by the adjacent nodes. 
In practice, however, some residential units (nodes) lie at the boundary of the municipality. In addition, some 
interior residential units  are separated by a physical barrier (river, highway) from adjacent units. 
Therefore, real two-dimensional networks are not periodic, but rather have exterior and possibly also interior boundaries.  

The goal of this paper is to study the role such boundaries play in the diffusion of new products. 
The paper is organized as follows. Section~\ref{sec:review} reviews the discrete Bass model. 
Section~\ref{sec:General-principles} introduces several  
 analytic tools that are used in the subsequent analysis. As in~\cite{OR-10}, {\em translation invariance} (Section~\ref{sec:translation}) simplifies the analysis on periodic networks. 
The  methodological contribution of this paper consists of two novel principles:
$(i)$~The {\em dominance principle} (Section~\ref{sec:dominance}) identifies pairs of networks for which the 
adoption in the first network is slower than in the second network, and $(ii)$~The {\em indifference principle} (Section~\ref{sec:Indifference}) enables the explicit calculation of the adoption curve on various networks, by identifying edges which have no effect on the adoption probabilities of certain sets of nodes. 
 
In Section~\ref{sec:circle} we  use the indifference principle 
to simplify the explicit calculations of one-sided and two-sided diffusion on the circle, which where first done in~\cite{OR-10}. 
In particular, we recover the result that {\em on the circle, one-sided and two-sided diffusion } coincide. In Sections~\ref{sec:line}
and~\ref{sec:hybrid}  we 
use the  indifference principle to obtain new results, namely, the
explicit calculation of one-sided and two-sided diffusion on a line, and on a hybrid network of a circle with a ray. 
In Section~\ref{sec:f_one-sided<f_two-sided} we prove our main result that {\em on the line, one-sided diffusion  
is strictly slower than two-sided diffusion}. Since one-sided and two-sided diffusion on the circle coincide, the difference between one-sided and two-sided diffusion on the line 
 is purely a boundary effect. This insight explains a previous finding~\cite{Bass-SIR-analysis-17} that when adopters are allowed to {\em recover} (i.e., to become non-contagious), 
one-sided diffusion on the circle is slower than two-sided diffusion. Indeed, once consumers begin to recover, the circle is broken into  several disjoint lines. Since one-sided diffusion on each of these lines is slower than two-sided diffusion, so is the overall diffusion.

In higher dimensions, the explicit calculation of the aggregate diffusion is an open problem. 
Numerical simulations suggest that the behavior is similar to the one in the one-dimensional case, namely, one-sided and two-sided diffusion are identical when the network is periodic, but that one-sided diffusion is strictly slower than two-sided diffusion when the network is non-periodic (Section~\ref{sec:D>=2}).  

\section{Discrete Bass model}
\label{sec:review}

We now review the discrete Bass model for diffusion of new products~\cite{Bass-69,OR-10}. 
A new product is introduced  at time $t=0$ to a market with $M$~potential consumers.
Initially all consumers are non-adopters.  
If consumer~$j$ adopts the product, she remains an adopter at all later times.
Let~$X_j(t)$ denote the state of consumer~$j$ at time~$t$, so that 
$$ 
 X_j(t)= \left\{ 
\begin{array}{ll}
  1, & \quad \mbox{if~$j$ adopted by time~$t$} ,\\
  0, & \quad  \mbox{otherwise}.
\end{array}
\right.
$$
The consumers belong to a social network which is described by an undirected or directed graph, such that
node $j$ has a positive weight $p_j$, and the directed edge \circled{$i$} $\to$ \circled{$j$} has a positive weight~$q_{i,j}$.\footnote{If the graph is undirected, then $q_{i,j}=q_{j,i}$.}
 If~$j$ did not adopt the product by time~$t$,
her probability to adopt it in the interval~$(t,t+\Delta t)$ is
\begin{equation}
  \label{eq:Bass-model-heterogeneous}
\text{Prob}{j~\text{adopts in}\choose {(t,t+\Delta t)}}= \Big( p_j + \sum_{i\not= j}q_{i,j} {X}_{i}(t) \Big)\Delta t, \qquad \Delta t \to 0. 
\end{equation}

The parameter~$p_j$ describes the likelihood of~$j$ to
adopt the product due to {\em external influences} by mass
media or commercials. Similarly, 
 the parameter~$q_{i,j}$ describe the likelihood of~$j$ adopting the product
due to {\em internal influences} ({\em word of mouth}, {\em peers effect}) by~$i$, provided that~$i$ had already adopted the product and that~$j$ can be influenced by $i$.\footnote{$q_{i,j}=0$ if there is no edge from $i$ to $j$.}  
The level of internal influences experienced by~$j$ is the sum of the individual influences of the adopters connected to~$j$. Typically, we assume that~$\sum_{i\neq j}q_{i,j}=q$ for all~$j$, i.e.\ the maximal internal effect experienced by~$j$ is~$q$~\cite{Bass-SIR-model-16}.

\subsection{Cartesian networks and boundary conditions}

In this study we mainly consider diffusion on $D$-dimensional Cartesian grids.
The two-dimensional case is relevant for products that spread predominantly through internal influences by geographical neighbors, 
such as residential rooftop solar systems~\cite{Bollinger-12,Graziano-15}. The one-dimensional case has the advantage that it can be solved explicitly, and is conjectured to serve as a lower bound for all other networks~\cite{OR-10}. 
 
We consider the homogeneous case where all nodes and all edges have the same weights, and
the adoption probability reads
\begin{equation}
\label{LinearAdoptionRates-new}
\text{Prob}{j~\text{adopts in}\choose {(t,t+\Delta t)}}= 
\left\{\begin{array}{lc}
\left(p+q \frac{n_j(t)}{k_D}\right) \Delta t+o(\Delta t), & \quad k_D >0,\\
  p\Delta t+o(\Delta t), & \quad k_D =0,
\end{array}
\right.
\end{equation}
as $\Delta t\to0$,  where~$n_j(t)$ is the number of adopters connected to~$j$ at time~$t$, and~$k_D$ is the 
number of consumers connected to~$j$ (the degree or in-degree of node~$j$).   


We consider both {\em two-sided diffusion}, where each node can be influenced by its $k_D = 2D$ nearest neighbors,  and {\em one-sided diffusion}, where each node can be influenced by its $k_D = D$ nearest neighbors.
The internal influence of  each adopter on a connected potential adopter is~$q/2D$ in the two-sided case,
and $q/D$ in the one-sided case. In the latter case, the diffusion is one-sided in each of the $D$~coordinates. 

For each of these two cases, we consider two types of boundary conditions:
\begin{description}
  \item[Periodic BC.]  
When we want to avoid boundary effects, we assume periodicity in each of the $D$~coordinates (i.e., the domain is a $D$-dimensional torus).
 In this case, all nodes have the same degree. 

\item[Non-periodic  BC.]
When we allow for boundary effects, the domain is a $D$-dimensional box~$B_D$. 
We can visualize this case as though we embed the box~$B_D$ in a larger domain, and set $X_j(t) \equiv 0$ for all nodes outside~$B_D$. 
Therefore, we effectively have a Dirichlet boundary condition at $\partial B_D$.\footnote{See, e.g., \eqref{eq:periodicity3} and~\eqref{eq:periodicity4}.} 

Since the degree of the boundary nodes can be smaller than that of the interior ones, 
and all edges have the same weight of~$q/2D$ in the two-sided case
and~$q/D$ in the one-sided case, see~\eqref{LinearAdoptionRates-new}, the maximal internal influence experienced by boundary nodes can be smaller than that of the interior ones.  
\end{description}

\subsection{1D networks}
\label{sec:1D-models}
  We describe the four network types  in the one-dimensional case with $M$ nodes.

\begin{figure}[ht!]
\begin{center}
\scalebox{0.6}{\includegraphics{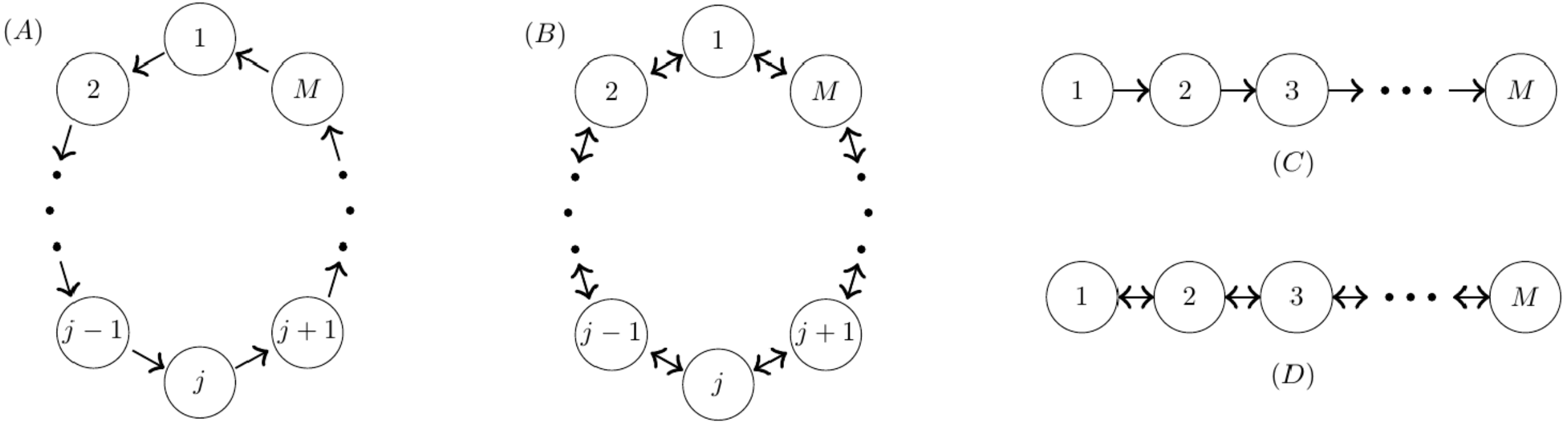}}
\caption{\emph{(}A\emph{)}~One-sided diffusion on a circle. \emph{(}B\emph{)}~Two-sided diffusion on a circle. \emph{(}C\emph{)}~One-sided diffusion on a line. \emph{(}D\emph{)}~Two-sided diffusion on a line.}
\label{fig:1D-graphs}
\end{center}
\end{figure}

\begin{description}
\item{\bf One-sided circle.}
Each consumer can only be influenced by her left neighbor 
 (Figure~\ref{fig:1D-graphs}A). 
Therefore, if~$j \in \{1, \dots, M\}$ has not yet adopted at time~$t$, then 
\begin{subequations}
\label{GeneralAdoptionRates_Left}
\begin{flalign}
\text{Prob}{j~\text{adopts in}\choose {(t,t+\Delta t)}}
= 
 \left(p+q X_{j-1}(t)\right)\Delta t+o(\Delta t),   
 \qquad \Delta t\to0, 
\end{flalign}
 where by periodicity
\begin{flalign}\label{eq:periodicity}
X_0(t):=X_M(t).
\end{flalign}
\end{subequations}

\item{\bf Two-sided circle.} Each consumer can be influenced by her left and right neighbors
 (Figure~\ref{fig:1D-graphs}B). 
Therefore, if~$j \in \{1, \dots, M\}$ has not yet adopted at time~$t$, then 
\begin{subequations}
\label{GeneralAdoptionRates_Left_Right}
\begin{flalign}
\text{Prob}{j~\text{adopts in}\choose {(t,t+\Delta t)}}
= \left(p+\frac{q}{2} \left( X_{j-1}(t) + X_{j+1}(t)\right)\right)\Delta t+o(\Delta t),  
 \qquad \Delta t\to0, 
\end{flalign}
 where by periodicity
\begin{flalign}\label{eq:periodicity2}
X_0(t):=X_M(t), \qquad X_{M+1}(t):=X_1(t).
\end{flalign}
\end{subequations}

\item{\bf One-sided line.} Each consumer can only be influenced by her left neighbor (Figure~\ref{fig:1D-graphs}C). Therefore, if~$j \in \{1, \dots, M\}$ has not yet adopted by time~$t$, then 
\begin{subequations}
\label{GeneralAdoptionRates_Left-line}
\begin{flalign}
\text{Prob}{j~\text{adopts in}\choose {(t,t+\Delta t)}}
= \left(p+q X_{j-1}(t)\right)\Delta t+o(\Delta t),  
 \qquad \Delta t\to0, 
\end{flalign}
 where 
\begin{flalign}
  \label{eq:periodicity3}
X_0(t)\equiv 0.
\end{flalign}
\end{subequations}

\item{\bf Two-sided line.} Each consumer can be influenced by her left and right neighbors
 (Figure~\ref{fig:1D-graphs}D). Therefore, 
if~$j \in \{1, \dots, M\}$ has not yet adopted at time~$t$, then 
\begin{subequations}
\label{eq:GeneralAdoptionRates_Left_Right-line}
\begin{flalign}
\text{Prob}{j~\text{adopts in}\choose {(t,t+\Delta t)}}
= \left(p+\frac{q}{2} \left( X_{j-1}(t) + X_{j+1}(t)\right)\right)\Delta t+o(\Delta t),  
 \qquad \Delta t\to0, 
\end{flalign}
where
\begin{flalign}\label{eq:periodicity4}
X_0(t)\equiv 0, \qquad X_{M+1}(t)\equiv 0.
\end{flalign}
\end{subequations}
\end{description}

\section{Analytic tools}
\label{sec:General-principles}

 Let us denote by $\Omega \subset \{1, \dots, M\}$ a subset of the $M$~nodes, 
by~$n(t)$ the number of adopters at time~$t$, and by $f(t)=\frac1M \mathbb{E}[n(t)]$ 
the expected fraction of adopters.
Then
\begin{equation}
\label{eq:f_sum}
   f(t) = \frac1M \sum_{j=1}^M \text{Prob} ( X_j(t)=1) = \frac1M \sum_{j=1}^M \mathbb{E}[X_j(t)] = \frac1M \sum_{j=1}^M \left(1-\text{Prob} ( X_j(t)=0)\right) .
\end{equation}



\subsection{Translation Invariance}
\label{sec:translation}

We can simplify the analysis on periodic Cartesian domains by utilizing translation invariance:

\begin{lemma}[Translation Invariance~\cite{OR-10}]
\label{lem:Translation-Invariance}
Let $M = m^D$. 
  Consider the homogeneous discrete Bass model~\eqref{LinearAdoptionRates-new} with one-sided or two-sided diffusion on a periodic hypercube~$[1,2,\dots ,m]^D$. Then the probability for adoption is the same for all nodes,
 i.e., $\emph{\text{Prob}} \emph{(} X_j(t)=1\emph{)}$ is independent of~$j$. Therefore, for any~$j$, 
$$
  f(t) = \emph{\text{Prob}} \emph{(} X_j(t)=1\emph{)} =  \mathbb{E}[X_j(t)] .
$$
\end{lemma}
More generally, the adoption probabilities of a set of nodes
are invariant under translation. 
For example, let us denote by~$S_k^{\text{one/two-sided}}(t;M)$ the probabilities that nodes $j+1,j+2,\ldots,j+k$ did not adopt by time $t$ in one-sided and two-sided circular networks with $M$ nodes, respectively. 
By translation invariance, $S_k^{\text{one/two-sided}}$ are independent 
of~$j$. Therefore, 
\begin{equation}
\label{eq:S_k-def}
S_k^{\text{one/two-sided}}(t;M) = \text{Prob}\left(X_{j+1}(t)=0,X_{j+2}(t)=0,\ldots,X_{j+k}(t)=0\right).
\end{equation} 
Obviously, translation invariance is lost in the non-periodic case.

Since $S_k^{\text{one-sided}}(t,M)\equiv S_k^{\text{two-sided}}(t,M)$, see~\cite{OR-10}, 
we sometimes drop the superscripts and denote
\begin{subequations}
\begin{equation}
\label{eq:sk_1s=sk_2s}
S_k^{\text{one-sided}}=S_k^{\text{two-sided}}=S_k.
\end{equation}
For $k=1$, we sometimes drop also the subscript and denote 
\begin{equation}
\label{eq:s1_1s=s1_2s}
S^{\text{one-sided}}_1=S^{\text{two-sided}}_1=S_1 = S.
\end{equation}
Since $S$ is the probability to be a non-adopter by time $t$ on (one-sided or two-sided) circular networks, it follows from Lemma~\ref{lem:Translation-Invariance} that
\begin{equation}
\label{eq:s=1-f}
S=1-f_{\text{circle}},
\end{equation}
where $f_{\text{circle}}$ is the expected fraction of adopters in (one-sided or two-sided) circular networks.\footnote{The expected fractions of adopters in one-sided and two-sided circular networks coincide, see~\cite{OR-10} and also~\eqref{eq:f=f1s=f2s}.}

\end{subequations}

\subsection{Dominance principle}
\label{sec:dominance}

	\begin{definition}
		Consider the heterogeneous Bass model~\eqref{eq:Bass-model-heterogeneous} on networks A and B with $M$~nodes, with external parameters $\{p_i^A\}$ and \{$p_i^B\}$, and with internal parameters  $\{q_{i,j}^{A}\}$ and $\{q_{i,j}^{B}\}$, respectively. 
  We say that $A \preceq B$ if 
$$
p_j^{A} \leq p_j^B \quad \mbox{for all~j} \qquad \mbox{and} \qquad  q_{i,j}^{A} \leq q_{i,j}^{B} \quad \mbox{for all~} i \not=j.
$$
 We say that $A \prec B$ if at least one of these $M^2$ inequalities is strict.
	\end{definition}

\begin{lemma}[Dominance principle]
\label{lem:dominance-principle}
	If $A \preceq B$ then $f_{A}(t) \leq f_{B}(t)$ for $t>0$. If $A \prec B$, then $f_{A}(t) < f_{B}(t)$ for $t>0$.
\end{lemma}
	\begin{proof}
		Assume first that $A \preceq B$. Let $t_{n} = n\Delta t$. For node $j$ in network $A$, define the random variable
		 \begin{align*}
		 X_{j}^{A}(t_{n}) = \begin{cases}
		 1, \quad &\text{if $j$ adopted by time } t_{n}, \\ 0, &\text{else},
		 \end{cases} \ \ \ \ \ &&j = 1,\dots, M, &&n=0,1,\ldots
		 \end{align*}
		 Let us define a specific realization $\widetilde{X}_{j}^{A}(t_{n})$ of $X_{j}^{A}(t_{n})$ as follows:
		 \renewcommand\labelitemii{$\bullet$} 		 
		 \renewcommand\labelitemiii{$\bullet$}
		 \renewcommand\labelitemiv{$\bullet$}  
		 \begin{itemize}
		 	\item $\widetilde{X}_{j}^{A}(0) = 0$ for $j=1,\ldots,M$
		 	\item for $n=1,2,\ldots$ \begin{itemize}
		 		\item sample a random vector $\vomega^{n} = (\omega_{1}^{n},\ldots,\omega_{M}^{n})$ from the uniform distribution on $[0,1]^{M}$
		 		\item for $j=1,\ldots,M$ \begin{itemize}
		 		\item if $\widetilde{X}_{j}^{A}(t_{n}) = 1$, then $\widetilde{X}_{j}^{A}(t_{n+1}) = 1$
		 		\item if $\widetilde{X}_{j}^{A}(t_{n}) = 0$, then
		 		\begin{itemize}
		 			\item if $\omega_{j}^{n} \leq \Big( p_j^{A} + \sum_{i\not= j}q_{i,j}^{A} \widetilde{X}_{i}^{A}(t_{n}) \Big)\Delta t$, then $\widetilde{X}_{j}^{A}(t_{n+1}) = 1$ 
		 			\item else $\widetilde{X}_{j}^{A}(t_{n+1}) = 0$ 
		 		\end{itemize}
		 		\end{itemize}
		 	\item end
		 	\end{itemize}
		 	\item end
		 \end{itemize}  
Define $X_{j}^{B}(t_{n})$ and $\widetilde{X}_{j}^{B}(t_{n})$ in the same manner. We claim that if we use the \textit{same} sequence $\{\vomega^{n}\}_{n=1}^{\infty}$ for $\widetilde{X}_{j}^{A}$ and $\widetilde{X}_{j}^{B}$, then
		 \begin{equation}
		 \label{eq:reduction}
		 \widetilde{X}_{j}^{A}(t_{n}) \leq \widetilde{X}_{j}^{B}(t_{n}), \qquad j=1,\ldots,M, \qquad  n=0,1,\ldots.
		 \end{equation}
		The result will follow from (\ref{eq:reduction}), because
		\begin{align*}
		\mathbb{E}\big[X_{j}^{A}(t_{n})\big] &= \int_{[0,1]^{M\times n}}\widetilde{X}_{j}^{A}\big(t_{n};\vomega^{1},\ldots,\vomega^{n}\big)d\vomega^{1}\cdots d\vomega^{n} \\ &\leq \int_{[0,1]^{M\times n}}\widetilde{X}_{j}^{B}\big(t_{n};\vomega^{1},\ldots,\vomega^{n}\big)d\vomega^{1}\cdots d\vomega^{n} = \mathbb{E}\big[X_{j}^{B}(t_{n})\big],
		\end{align*}
		and so
		$
		f_{A}(t_{n}) = \frac{1}{M}\sum_{j=1}^{M}\mathbb{E}\big[X_{j}^{A}(t_{n})\big] \leq \frac{1}{M}\sum_{j=1}^{M}\mathbb{E}\big[X_{j}^{B}(t_{n})\big] = f_{B}(t_{n}).
		$
        
		We prove (\ref{eq:reduction}) by induction on $n$. For $n=0$, (\ref{eq:reduction}) holds since $\widetilde{X}_{j}^{A}(0) = \widetilde{X}_{j}^{B}(0) = 0$. To prove the induction step, we only need to consider the case $\widetilde{X}_{j}^{A}(t_{n}) = \widetilde{X}_{j}^{B}(t_{n}) = 0$. 
Now
\begin{align*}
			\bigg( p_j^{A} + \sum_{i\not= j}q_{i,j}^{A} \widetilde{X}_{i}^{A}(t_{n}) \bigg)\Delta t 
\leq 
 \bigg( p_j^{A} + \sum_{i\not= j}q_{i,j}^{A} \widetilde{X}_{i}^{B}(t_{n}) \bigg)\Delta t 
 &\leq \bigg( p_j^{B} + \sum_{i\not= j}q_{i,j}^{B} \widetilde{X}_{i}^{B}(t_{n}) \bigg)\Delta t,
		\end{align*}
		where the first inequality follows from the induction assumption. Hence if $\widetilde{X}_{j}^{A}(t_{n+1}) = 1$, then 
$\omega_{j}^{n}\leq\Big( p_j^{A} + \sum_{i\not= j}q_{i,j}^{A} \widetilde{X}_{i}^{A}(t_{n}) \Big)\Delta t \leq \Big( p_j^{B} + \sum_{i\not= j}q_{i,j}^{B} \widetilde{X}_{i}^{B}(t_{n}) \Big)\Delta t$, and so $\widetilde{X}_{j}^{B}(t_{n+1}) = 1$ as well.
		
 The extension of the proof for the case where $A\prec B$ goes as follows. It is easy to verify that for some sequences $\{\vomega^{n}\}_{n=1}^{\infty}$,  nodes~$j$ in~$B$ and in~$A$ adopt at the same time, while for other sequences (that have a positive measure) node~$j$ in~$B$ adopts strictly before node~$j$ in~$A$. There are, however, no sequences 
$\{\vomega^{n}\}_{n=1}^{\infty}$ for which node~$j$ in~$A$ adopts before node~$j$ in~$B$.
	\end{proof}

An immediate consequence of the dominance principle is
\begin{corollary}
\label{cor:f_a<f_B}
If network~$B$ is obtained from network~$A$ by adding links with positive weights, then $f_A(t) < f_B(t)$ for $t > 0$.
\end{corollary}

 We can generalize the dominance principle to subsets of the nodes.
\begin{lemma}[Generalized dominance principle]
\label{lem:dominance-principle-generalized}
Let $S_{\rm \Omega}^{A,B}(t)$ denote the probabilities that none of the nodes in~$\Omega$ have adopted by time~$t$ in networks $A$ and $B$, respectively, where $\Omega \subset \{1, \dots, M\}$.
	If $A \preceq B$, then $S_{\rm \Omega}^A(t) \ge  S_{\rm \Omega}^B(t)$ for $t>0$. 
\end{lemma}
\begin{proof}
 Let $t_\Omega=t_\Omega(\{\vomega^{n}\}_{n=1}^{\infty})$ denote the time of the first adoption in~$\Omega$ under sequence $\{\vomega^{n}\}_{n=1}^{\infty}$. 
Then
\begin{equation}
\label{eq:S_indicate}
S_{\Omega}(t_n) = 
\int\displaylimits_{~~[0,1]^{M\times n}}    \!\!\!\!\! \!\!\!\!\! \big[\mathbbm{1}_{\{t_\Omega>t_n\}}\big]d\vomega^{1}\cdots d\vomega^{n}, \qquad \mathbbm{1}_{\{t_\Omega>t_n\}}:=
\begin{cases}
1, \qquad \text{if}~t_\Omega>t_n,\\
0, \qquad \text{otherwise}.
\end{cases}
\end{equation}
By~(\ref{eq:reduction}), $t_\Omega^A \ge t_\Omega^B$.  Therefore, the result follows.
\end{proof}

It is not true, however, that if $A \prec B$, then $S_{\rm \Omega}^A(t) >  S_{\rm \Omega}^B(t)$ for $t>0$. Indeed, if we only change the weights of {\em non-influential edges to~$\Omega$} (see Section~\ref{sec:Indifference}),  this will have no effect on~$S_\Omega$. We can prove a strict inequality, however,   
if we increase the weights of {\em influential edges}. For example, we have
\begin{lemma}
\label{lem:s2<s1}
Let $S_k$ be given by~\eqref{eq:S_k-def}. Then
$$
S^{\text{\emph{one/two-sided}}}_k(t;p,q_2,M)<S^{\text{\emph{one/two-sided}}}_k(t;p,q_1,M), \qquad q_2>q_1, \quad t>0.
$$
\end{lemma}
\begin{proof}
The proof is similar to that of Lemmas~\ref{lem:dominance-principle} and~\ref{lem:dominance-principle-generalized}.
%
\end{proof}

\subsection{Indifference principle}
\label{sec:Indifference}

\begin{definition}[influential and non-influential edges]
\label{def:influential-edge}
Consider a directed network with $M$~nodes \emph{(}if the network is undirected, replace each undirected edge  by two directed edges\emph{)}.
Let $\Omega \subsetneqq \{1, \dots, M\}$ be a subset of the nodes, and let $\Omega^{\rm \emph{c}} = \{1, \dots, M\} \setminus \Omega $ be its complement.
A directed edge \circled{$a$}~$\to$~\circled{$b$} is called ``non-influential to $\Omega$'' if
\begin{enumerate}
  \item  $a \in \Omega$, or 
  \item  $a  \in \Omega^{\rm \emph{c}}, b  \in \Omega^{\rm \emph{c}}$, and there is no sequence of directed edges from $b$ to $\Omega$, or
  \item  $a  \in \Omega^{\rm \emph{c}}, b  \in \Omega^{\rm \emph{c}}$, and all sequences of directed edges from $b$ to $\Omega$ go through the node $a$.
\end{enumerate} 
An edge which is not  non-influential to $\Omega$ is called ``influential to $\Omega$".\footnote{Thus, a directed edge \circled{$a$}~$\to$~\circled{$b$} is influential to $\Omega$ if: 1.~$a \in \Omega^{\emph{\text{c}}}$, and 2.~either $b \in \Omega$, or there is a sequence of directed edges from $b$ to a node $u \in \Omega$ which does not go thorough the node $a$. } 
\end{definition}

An illustration of influential and non-influential edges is shown in Figure~\ref{fig:influential}.

\begin{figure}[ht!]
\begin{center}
\scalebox{1}{\includegraphics{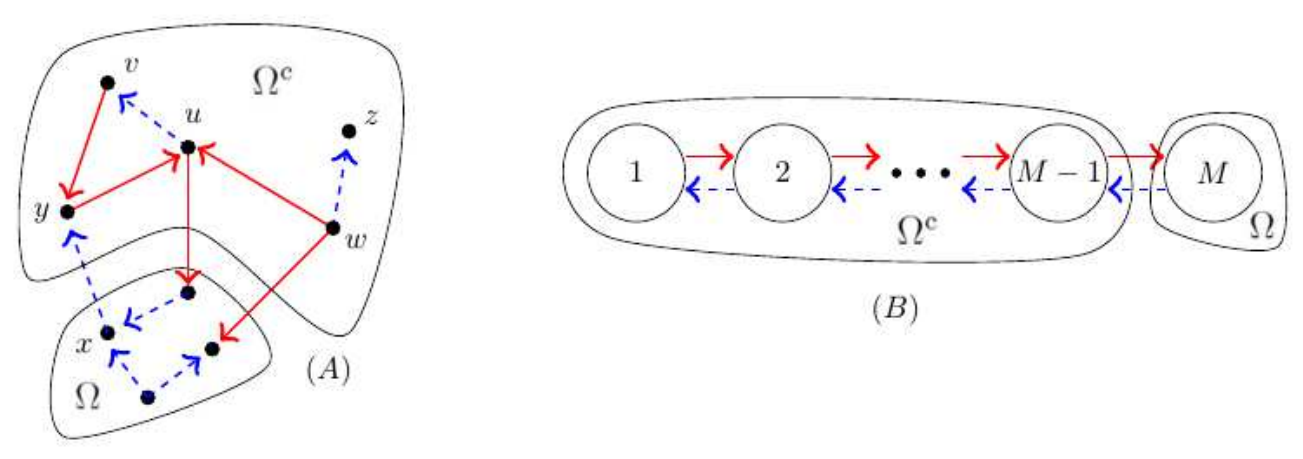}}
\caption{Illustrations of influential and non-influential edges. The nodes are divided into two separate sets: $\Omega$ and $\Omega^{\rm \emph{c}}$. Influential and non-influential edges to~$\Omega$ are denoted by solid red and dashed blue arrows, respectively. The non-influential edges \protect\circled{$x$} $\to$ \protect\circled{$y$}, \protect\circled{$w$} $\to$ \protect\circled{$z$} and \protect\circled{$u$} $\to$ \protect\circled{$v$} in \emph{(}A\emph{)} correspond to cases 1,2 and 3 in \emph{Definition} \emph{\ref{def:influential-edge}}, respectively.}
\label{fig:influential}
\end{center}
\end{figure}

\begin{lemma}[Indifference principle]
\label{lem:Indifference}
Let $S_{\rm \Omega}(t)$ denote the probability that all the nodes in~$\Omega$ did not adopt by time~$t$ 
in the discrete Bass model~\eqref{eq:Bass-model-heterogeneous}. Then $S_{\rm \Omega}(t)$
remains unchanged if we remove or add non-influential edges to~$\Omega$.
\end{lemma}

\begin{proof}
We start with two identical networks $A$ and $B$ that have the same nodes $\{1,\dots,M\}$, the same external parameters $\{p_i\}$, and the same internal parameters $\{q_{i,j}\}$.  Then we add and remove edges in network $B$ that are non-influentials to $\Omega$. 

As in the proof of the dominance principle, we consider two specific realizations $\widetilde{X}^A_j(t_n)$ and $\widetilde{X}^B_j(t_n)$ which are produced from the same sequence $\{\vomega^{n}\}_{n=1}^{\infty}$. Let $j^A$ and $j^B$ denote the first node in $\Omega$ that adopts in networks $A$ and $B$, respectively, and let $t^A$ and $t^B$ be the times at which these adoptions occur. We claim that
\begin{equation}
\label{eq:indifference_AB}
j^A = j^B, \qquad t^A = t^B.
\end{equation} 
The result will follow from~\eqref{eq:S_indicate} and~\eqref{eq:indifference_AB}.

We prove \eqref{eq:indifference_AB} by contradiction. Assume by negation that $t^B<t^A$.
Then the external influence on node $j^{B}$ at time $t^B-\Delta t$ in $B$ is greater than in $A$. Let $t^{1}<t^B$ be the earliest time where the external influence on $j^{B}$ was greater in $B$ than in $A$. Then at $t^{1-}:=t^{1}-\Delta t$, some node $j_1^{B}$ in $B$ which has an influence over node $j^{B}$ decided to adopt. Since until $t^B$ all of the nodes in $\Omega$ are non-adopters in $B$, we have that
$
j_1^{B}\in \Omega^c.
$
Since edge \circled{$j_1^{B}$} $\to$ \circled{$j^{B}$} is influential in $B$, and no influential edges were added, we conclude that in $A$, node $j^B_1$ also has an influence over node $j^B$. Hence, at time $t^{1-}$ node $j_1^{B}$ in $A$ remains a non-adopter.

We now consider the scenario that at time $t^{1-}$, node $j_1^{B}$ decided to adopt in $B$, but remained a non-adopter in~$A$.
By repeating the arguments of the previous stage, we deduce that this is possible only if at some time $t^{2-}<t^{1-}$, some node $j^B_2$ in $B$ which has an influence over $j^B_1$ decided to adopt. We recall that until time $t^B$ all of the nodes in $\Omega$ are non-adopters in $B$. In addition, until time $t^{1-}$ node $j^B_1$ in $B$ is also a non-adopter. This implies that
\begin{equation*}
j^B_2\in \{\Omega_1\}^c, \qquad \Omega_1 := \Omega\cup j^B_1.
\end{equation*}
By construction, nodes $j^B$, $j^B_1$ and $j^B_2$ are all distinct from one another.
Combining the facts that for network $B$ we have that $j_2^B\in \Omega^c$, that the path \circled{$j_2^B$} $\to$ \circled{$j_1^B$} $\to$ \circled{$j^B$} consists of distinct nodes only, 
and that no influential edges were added to $B$, we get that in $A$, node $j_2^{B}$ also has an influence over node $j_1^{B}$. Hence, at time $t^{2-}$ node $j_2^{B}$ in $A$ remains a non-adopter.

By repeating the above argument, we obtain sequences of nodes $\{j^{B}_k\}$, sets $\{\Omega_k\}$, and times $\{t^{k-}\}$, for $k=1,2,\ldots$. Since $k$ is arbitrarily large, and the sequence of sets $\Omega_1 \subsetneqq \Omega_2\subsetneqq \cdots \subsetneqq\Omega_k$ is strictly increasing, but the number of nodes $M$ is finite, we get a contradiction. In the case of an infinite network, the contradiction can be obtained by observing that $t^B$ is reached after a finite number of discrete time steps, but $t^{1-}>t^{2-}>t^{3-}>\cdots>t^{k-}$, and so $t^{k-}$ becomes negative for $k$ sufficiently large.  
\end{proof}

An immediate consequence of the {\em indifference principle} is
\begin{corollary}  
$S_{\rm \Omega}(t)$ remains unchanged if nodes in $\Omega$ change their influences on other nodes \emph{(}both inside and outside of~$\Omega$\emph{)}.
\end{corollary}
\begin{proof}
Any directed edge that starts from a node in $\Omega$ is non-influential to $\Omega$. Therefore, the result follows from Lemma~\ref{lem:Indifference}.
\end{proof}

To motivate condition 3 in Definition~\ref{def:influential-edge} for a non-influential edge, we first note that the sequence of influential edges in the proof of the indifference principle does not go through the same node more than once. To further motivate this condition, consider the edge $\circled{u} \to \circled{v}$ in Figure~\ref{fig:influential}A. 
This edge is non-influential to~$\Omega$, because for it to influence~$\Omega$, $u$ should be an adopter in order to influence $v$. But then, for $v$ to influence $\Omega$,
$v$ should influence $u$. But $u$ is already an adopter, so the edge $\circled{u} \to \circled{v}$ cannot influence $\Omega$.
A second example is Figure \ref{fig:influential}B, where $\Omega$ is node $M$. All the left-going edges  $\circled{k+1} \to \circled{k}$ are non-influentials, because if $k+1$ adopted, then its influence on $k$ has no effect on the future adoption of~$M$, since to get back from $k$ to $M$, one has to go through $k+1$, but  $k+1$ is already an adopter.

We now present two applications of the indifference principle. Additional applications are given in 
Lemmas~\ref{lem:f_one-sided_line},  \ref{lem:f_two-sided_line}, \ref{lem:circle_with_line}, and~\ref{lem:X1XM} and in Appendix \ref{app:psi}.
\begin{lemma}  
\label{lem:S_k^one-sided}
Consider the Bass model~\eqref{GeneralAdoptionRates_Left} on a one-sided circle. Then 
\begin{equation}
\label{indifference_2}
S_k^{\rm one-sided}(t;M) = S^{\rm one-sided}_{1}(t;M-(k-1))e^{-(k-1)pt}, \qquad k=2, \dots, M,
\end{equation}
where $S_k^{\rm one-sided}$ is given by \eqref{eq:S_k-def}.\footnote{$S_1^{\rm one-sided}(t;\widetilde{M})$ is the probability of a node in a one-sided circle with $\widetilde{M}$ nodes to remain a non-adopter.}   
\end{lemma}
\begin{proof}
By translation invariance, $S_k^{\rm one-sided}$ is independent of~$j$.
By the indifference principle, we can calculate $S_k^{\rm one-sided}$ from the network illustrated in Figure~\ref{fig:indifference_s_k}B. In that network, the states of $j+1,\dots,j+k$ are independent, and so
\begin{eqnarray*}
 S_k^{\rm one-sided}(t; M) &=& \prod_{m=j+1}^{j+k} \text{Prob}\left(X_{m}(t)=0\right).
\end{eqnarray*}
For $ m = j+1$ we have that $\text{Prob}\left(X_{m}(t)=0\right) =  S^{\rm one-sided}(t;M-k+1)$.
For $ j+2\le m \le j+k$, since~$m$ is not influenced by other individuals, $\text{Prob}\left(X_{m}(t)=0\right) = S(t;M=1)$.  

The equation for $S(t;M=1)$ reads 
$\frac{\text{d}}{\text{dt}}S(t;M=1) =-pS(t;M=1)$, see~\cite[Eq.~$(16)$]{OR-10}, subject to $S(t=0;M=1)=1$. Therefore,
\begin{equation}
  \label{eq:S_M=1}
S(t;M=1) = e^{-pt}.
\end{equation}
Hence, the result follows. 
\end{proof}

\begin{figure}[ht!]
\begin{center}
\scalebox{0.75}{\includegraphics{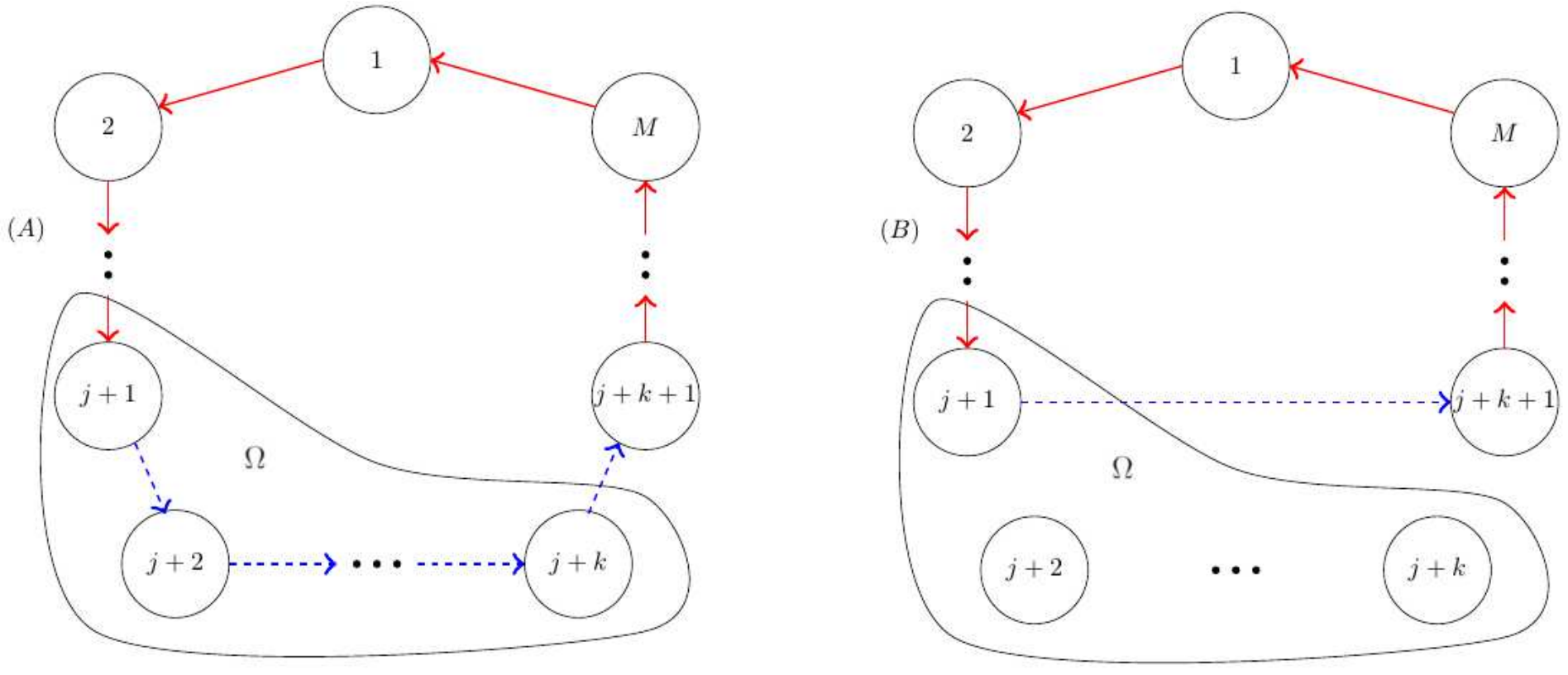}}
\caption{Equivalent networks for the calculation of $S_k^{\rm one-sided}(t;M)$ which are used in the proof of Lemma~\emph{\ref{lem:S_k^one-sided}}. Solid red and dashed blue arrows correspond to influential and non-influential edges to $\Omega = \{j+1,j+2,\dots,j+k\}$, respectively. \emph{(}A\emph{)}~One-sided circle. \emph{(}B\emph{)}~The $k-1$ non-influential edges 
\protect\circled{$j+1$} $\to \dots \to$ \protect\circled{$j+k$} are deleted. The non-influential edge \protect\circled{$j+1$} $\to$ \protect\circled{$j+k+1$} is added. }
\label{fig:indifference_s_k}
\end{center}
\end{figure}

\begin{lemma} 
\label{lem:S_k^two-sided}
 Consider the Bass model~\eqref{GeneralAdoptionRates_Left_Right} on a two-sided circle.
Then 
\begin{equation}
\label{eq:indifference_s_k_two_sided}
S_k^{\rm two-sided}(t;M) = S_2^{\rm two-sided}(t;M-(k-2))e^{-(k-2)pt}, \qquad k=3, \dots, M,
\end{equation}
where $S_k^{\rm two-sided}$ is given by \eqref{eq:S_k-def}.
\end{lemma}
\begin{proof}
By translation invariance, $S_k^{\rm two-sided}$ is independent of~$j$.
By the indifference principle, we can calculate $S_k^{\rm two-sided}(t;M)$ from the network
 illustrated in Figure~\ref{fig:indifference_s_k_two_sided}B. In that network, the states of the two-node set $\{j+1,j+k\}$ and the single-node sets 
$\{j+2\},\dots,\{j+k-1\}$ are independent, and so
\begin{eqnarray*}
 S_k^{\rm two-sided}(t)  &=& 
 \text{Prob}\left(X_{j+1}(t)=0,X_{j+k}(t)=0\right) \prod_{m=j+2}^{j+k-1} \text{Prob}\left(X_{m}(t)=0\right).
\end{eqnarray*}
Since $\text{Prob}\left(X_{j+1}(t)=0,X_{j+k}(t)=0\right) = S_2^{\rm two-sided}(t;M-k+2)$ and 
 $\text{Prob}\left(X_{m}(t)=0\right)=e^{-pt}$ for $m=j+2, \dots, j+k-1$, see~\eqref{eq:S_M=1}, the result follows.
\end{proof}

\begin{figure}[ht!]
\begin{center}
\scalebox{0.75}{\includegraphics{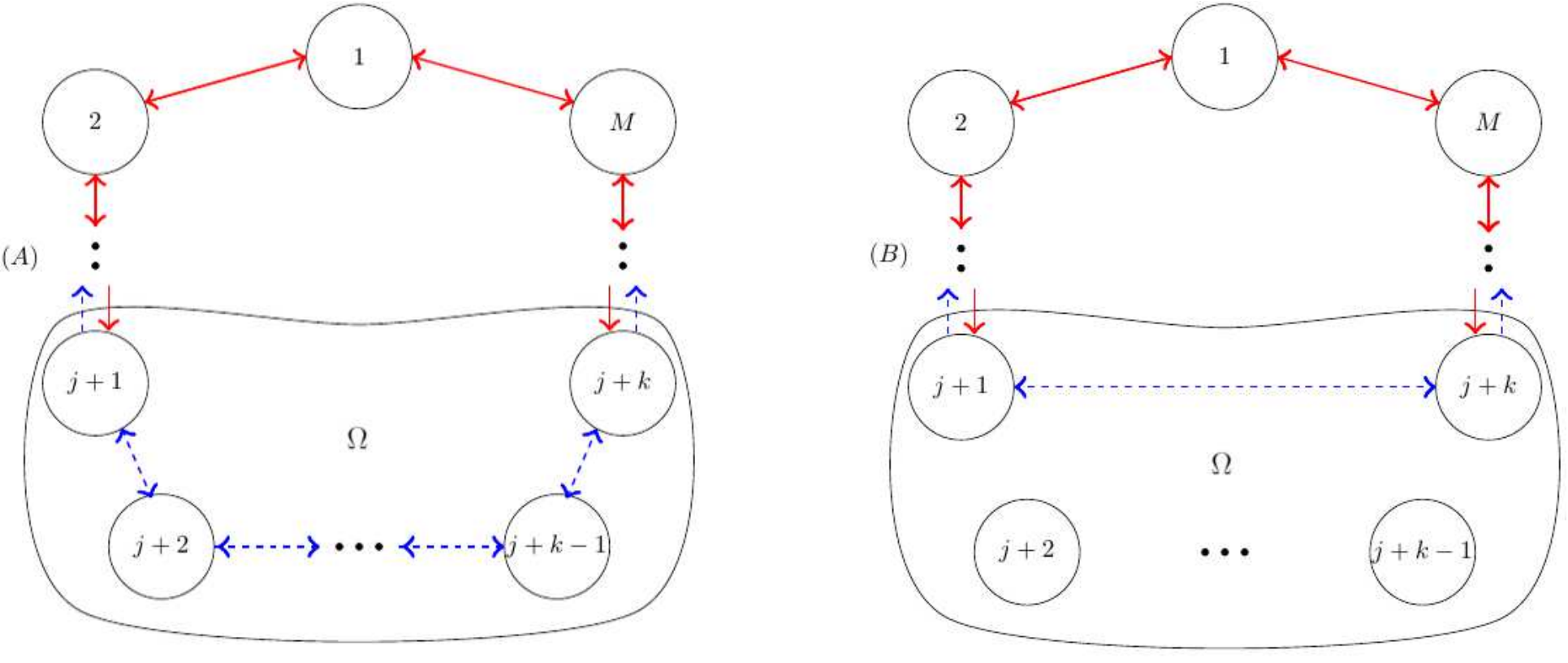}}
\caption{Equivalent networks for the calculation of $S_k^{\rm two-sided }(t;M)$, which are used in the proof of Lemma~\emph{\ref{lem:S_k^two-sided}}. Solid red and dashed blue arrows correspond to influential and non-influential edges to $\Omega = \{j+1,j+2,\dots,j+k\}$, respectively. \emph{(}A\emph{)}~Two-sided circle. \emph{(}B\emph{)} {The $k-1$ non-influential edges 
{ \protect\circled{$j+1$} $\leftrightarrow \dots \leftrightarrow$ \protect\circled{$j+k$}} are deleted. The non-influential edge \protect\circled{$j+1$} $\leftrightarrow$ \protect\circled{$j+k$}} is added.} 
\label{fig:indifference_s_k_two_sided}
\end{center}
\end{figure}



\Remark By Lemma~\ref{lem:S_k^one-sided}, $S_2^{\rm one-sided}(t;M-(k-2)) = S^{\rm one-sided}_{1}(t;M-(k-1))e^{-pt}$, and so
$$
S_k^{\rm one-sided}(t;M) = S_2^{\rm one-sided}(t;M-(k-2))e^{-(k-2)pt}.
$$
 This relation is also satisfied by $S_k^{\rm two-sided}$, see~\eqref{eq:indifference_s_k_two_sided}. 
Indeed, in~\cite{OR-10},  Fibich and Gibori showed that $S_k^{\rm one-sided}(t;M) \equiv S_k^{\rm two-sided}(t;M)$ for  $1\leq k \leq M$.
Therefore, we sometimes drop the superscripts {\em one-sided} and {\em two-sided}, and denote
$S_k^{\rm one-sided}$ and~$S_k^{\rm two-sided}$ by~$S_k$, see~\eqref{eq:sk_1s=sk_2s}. 

\section{Diffusion in 1D networks}
\label{sec:1D-networks}

In this section we use the indifference principle to explicitly calculate the diffusion in one-dimensional networks. 

\subsection{Periodic case (circle)}
\label{sec:circle}

 We begin with the one-sided circle:
\begin{lemma}[\cite{OR-10}]
   \label{lem:1D_Left_Explicit_prop}
   Let $q \not=p, 2p, \dots, (M-1)p$. Then the expected fraction of adopters on the one-sided circle with $M$ nodes, see~\eqref{GeneralAdoptionRates_Left},  is
    \begin{subequations}
    \label{1D_Left_Explicit-periodic}
  \begin{equation}
    \label{1D_Left_Explicit-periodic-A}
 f^{\rm one-sided}_{\rm circle}(t;p,q,M) =1-\sum_{k=1}^{M-1} c_k
 \frac{(-q)^{k-1}}{p^{k-1}(k-1)!} e^{(-kp-q)t}  -  
 \frac{(-q)^{M-1}}{\prod_{j=1}^{M-1}(jp-q)} e^{-Mpt},
 \end{equation}
 where  
 \begin{equation}
 \label{c_m-k}
 c_{M-k} =  1-\frac{q^k}{\prod_{j=1}^{k}(q-jp)}-\sum_{j=1}^{k-1}\frac{p^{j-k}(-q)^{-j+k}}{(k-j)!}c_{M-j}, \qquad k=1, \dots, M-1.
 \end{equation}
 \end{subequations}
\end{lemma}
\begin{proof}
This result was originally proved in~\cite{OR-10}. Here we provide a simpler proof, which illustrates the power and beauty of the indifference principle.
 Let $S_k(t;M)$
denote the probability that $k$ adjacent nodes remained non-adopters by time~$t$ in a circle with $M$~nodes, 
see \eqref{eq:S_k-def} and~\eqref{eq:sk_1s=sk_2s}.
 In~\cite{OR-10}, it was shown that
$f = 1-S_1$, where $S_1$ satisfies 
\begin{equation}
\label{evolution_eq}
S_1'(t;M) = -(p+q)S_1(t;M)+qS_2(t;M), \qquad S_1(0) = 1.
\end{equation}
Thus, $S_1'$ depends on $S_2$. Similarly, 
$S_2'$ depends on~$S_3$, etc. Therefore, to close the system in~\cite{OR-10}, Fibich and Gibori derived 
the following system of ODEs for $\{S_k(t; M)\}_{k=1}^M$:\footnote{This system holds for both one-sided and two-sided 
diffusion~\cite{OR-10}.}
\begin{subequations}
\label{eq:sk_deriviative}
\begin{flalign}
 S_k'(t;M) = (-kp-q)S_k(t;M) +qS_{k+1}(t;M), \quad S_k(0) = 1, \qquad  k=1,\dots, M-1,
\end{flalign}
\begin{flalign}
 S_M'(t;M) =&-MpS_M(t;M), \qquad S_M(0) = 1.
\label{eq:sM_deriviative}
\end{flalign}
\end{subequations}

Here we take a different approach, and 
close equation~\eqref{evolution_eq} using the relation 
\begin{equation}
\label{eq:indifference_2}
S_2(t;M) = S(t;M-1)e^{-pt},
\end{equation}
 see~\eqref{indifference_2} and \eqref{eq:s1_1s=s1_2s}, which was derived using the indifference principle.
Combining \eqref{evolution_eq} and \eqref{eq:indifference_2} gives 
\begin{equation}
\label{eq:final_form-ODE}
S'(t;M) +(p+q)S(t;M)=qe^{-pt}S(t;M-1), \qquad S(0;M) = 1, \qquad M=2,3, \dots 
\end{equation}
 The solution of this first-order linear ODE reads
\begin{equation}
\label{final_form}
S(t;M) = e^{-(p+q)t}+qe^{-(p+q)t}\int_0^t\left[e^{q\tau}S(\tau;M-1)\right]d\tau  , \qquad M=2,3, \dots 
\end{equation}
This recursion relation expresses $S(t;M)$ in terms of $S(t;M-1)$. For example, substituting $S(t;M=1)=e^{-pt}$, see~\eqref{eq:S_M=1}, in~\eqref{final_form} yields for $q \not=p$, that
$
S(t;M=2) = \left(1-\frac{q}{q-p}\right)e^{-(p+q)t}+\frac{q}{q-p}e^{-2pt}.
$
This, in turn, can be substituted in~\eqref{final_form}, yielding for $q \not=p, 2p$, that
$
S(t;M=3) = \left(1-\frac{q}{q-p}-\frac{q^2}{(q-p)(q-2p)}\right)e^{-(p+q)t}+\frac{q}{q-p}e^{-(2p+q)t}+\frac{q^2}{(q-p)(q-2p)}e^{-3pt}.
$
More generally, it follows by induction from~\eqref{final_form} that for $q \not=p, 2p, \dots,  (M-1)p$,
\begin{equation}
   \label{eq:SM}
 S(t;M)=\sum_{k=1}^{M-1} A_{k,M}e^{-(kp+q)t}+B_{M}e^{-Mpt}, \qquad M=1,2, \dots,  
\end{equation}
where $\{ A_{k,M} \}_{k=1}^{M-1}$ and $B_{M}$ are constants that depend on $p$, $q$, and~$M$. Substituting~\eqref{eq:SM} 
in both sides of~\eqref{final_form}, integrating the right-hand side terms, and equating the coefficients of the exponentials on both sides, 
gives the result (see Appendix~\ref{app:1D_Left_Explicit_prop}).
\end{proof}

The explicit expression~\eqref{1D_Left_Explicit-periodic} for the adoption curve simplifies as $M \to \infty$: 
\begin{lemma}[\cite{OR-10}]
   \label{lem:1D_Left_Explicit_prop-M=infinity}
\begin{equation}
    \label{eq:f_1D}
\lim_{M \to \infty} f^{\rm one-sided}_{\rm circle}(t;p,q,M)=f_{\rm 1D}(t;p,q), \qquad f_{\rm 1D}(t;p,q):=  1-
e^{-(p+q)t +\frac{q}{p} (1-e^{-pt})}.
\end{equation}
\end{lemma}
\begin{proof}
This result was originally proved in~\cite{OR-10}. Here we use the indifference principle to provide a simpler proof. Letting $M \to \infty$ in~\eqref{eq:final_form-ODE} gives
$$
S'(t;M=\infty) +(p+q)S(t;M=\infty)=qe^{-pt}S(t;M=\infty), \qquad S(0;M=\infty) = 1.
$$
Solving this linear first-order ODE and substituting $f=1-S$ gives~\eqref{eq:f_1D}. 
\end{proof}

%

Next, we consider the two-sided circle case.

\begin{lemma}[\cite{OR-10}]
  \label{lem:circle_two-sided}
The expected fraction of adopters on the two-sided circle with $M$~nodes, see~\eqref{GeneralAdoptionRates_Left_Right}, is identical to that on the one-sided circle, i.e., 
\begin{equation}
\label{eq:f_2sided_circle=f_1sided_circle}
f^{\rm two-sided}_{\rm circle}(t;p,q,M) =
f^{\rm one-sided}_{\rm circle}(t;p,q,M),  
\end{equation}
where $f^{\rm one-sided}_{\rm circle}$ is given by~\eqref{1D_Left_Explicit-periodic}. In particular, 
\begin{equation}
\label{eq:f_2sided_circle=f_1sided_circle_M2infty}
  \lim_{M \to \infty}f^{\rm two-sided}_{\rm circle}(t;p,q,M) = f_{\rm 1D}(t;p,q).
\end{equation}
\end{lemma}
\begin{proof}
This result was originally proved in~\cite{OR-10}. Here we again provide a different proof which makes use of the indifference principle. We recall that 
in~\cite{OR-10} it was shown (for both the one-sided and two-sided cases) that $S'(t;M)$ is given by~\eqref{evolution_eq}, and that
\begin{equation}
\label{eq:s2}
S_2'(t;M) = -(2p+q)S_2(t;M)+qS_3(t;M).
\end{equation}
 To close the ODEs system in the two-sided case, we use the 
indifference principle to get 
$
S_3(t;M) = S_2(t;M-1)e^{-pt},
$
 see~\eqref{eq:indifference_s_k_two_sided}. Plugging this in~\eqref{eq:s2} yields 
\begin{equation}
 \label{eq:S_2'}
S_2'(t;M) = -(2p+q)S_2(t;M)+qS_2(t;M-1)e^{-pt},\qquad S_2(0;M) = 1 
\end{equation}
for $M\geq 3$. In addition, the equation for $S_2(t;M=2)$ reads 
$\frac{\text{d}}{\text{dt}}S_2(t;M=2)=-2pS_2(t;M=2)$, subject to $S_2(0,M=2)=1$, and so
$$
S_2(t;M=2)=e^{-2pt}.
$$
Therefore, if we substitute $S_2(t;M)=e^{-pt}g(t;M-1)$ in~\eqref{eq:S_2'}, we get that $g(t;M)$
satisfies the same recursion relation as $S^{\rm one-sided}(t;M)$, see~\eqref{eq:final_form-ODE}. 
In addition,
$$
 g(t;M=1)=e^{pt} S_2(t;M=2)=  e^{-pt}=S^{\rm one-sided}(t;M=1).
$$
 Therefore, it follows that  
 $g(t;M) = S^{\rm one-sided}(t;M)$.  Hence, 
\begin{equation}
   \label{eq:S2_two_sided=S2_one_sided}
S^{\rm two-sided}_2(t;M)= e^{-pt}g(t;M-1) = e^{-pt}S^{\rm one-sided}(t;M-1) = 
S^{\rm one-sided}_2(t;M),
\end{equation}
where in the last equality we used~\eqref{indifference_2}.
Since for both the one-sided and two-sided cases, $S'(t;M)$ is given by~\eqref{evolution_eq}, 
it follows from~\eqref{eq:S2_two_sided=S2_one_sided} that $S^{\rm two-sided}(t;M)=S^{\rm one-sided}(t;M)$, and so \eqref{eq:f_2sided_circle=f_1sided_circle} follows. The limit~\eqref{eq:f_2sided_circle=f_1sided_circle_M2infty} follows from~\eqref{eq:f_2sided_circle=f_1sided_circle} and~\eqref{eq:f_1D}.
\end{proof}

Thus, the aggregate diffusions on the one-sided and two-sided circle are identical, as is confirmed numerically in 
Figure~\ref{fig:cellular_1D_one_versus_two_sided_main}A.
 Therefore, from now on we drop the superscripts and denote 
\begin{equation}
\label{eq:f=f1s=f2s}
f^{\rm one-sided}_{\rm circle}=f^{\rm two-sided}_{\rm circle}=f_{\rm circle}.
\end{equation} 

\begin{figure}[ht!]
\begin{center}
\scalebox{0.7}{\includegraphics{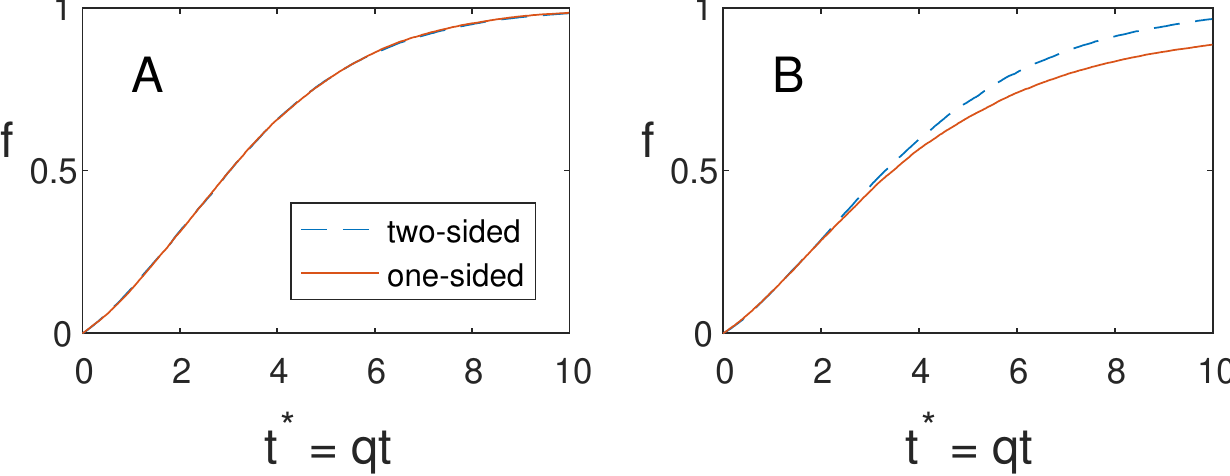}}
\caption{Fractional adoption on a one-dimensional network with $M = 6$ nodes, averaged over $\emph{4000}$ simulations, under two-sided \emph{(}dashed blue line\emph{)} and one-sided \emph{(}solid red line\emph{)} 
diffusion. Here $p=0.01$ and $q=0.1$. \emph{(}A\emph{)} Periodic boundary conditions \emph{(}circle\emph{)}. \emph{(}B\emph{)} Non-periodic boundary conditions \emph{(}line\emph{)}. }
\label{fig:cellular_1D_one_versus_two_sided_main}
\end{center}
\end{figure}

\subsection{Non-periodic case (line)}
\label{sec:line}

We now use the indifference principle to derive an explicit expression for the adoption curve
on the  one-sided line:
\begin{lemma}
\label{lem:f_one-sided_line}
The expected fraction of adopters on the one-sided line with $M$~nodes, see~\eqref{GeneralAdoptionRates_Left-line}, is given by 
$$
f^{\rm one-sided}_{\rm line}(t;p,q,M) :=
      \frac1M \sum_{j=1}^M f_j^{\rm one-sided}(t;p,q,M)=\frac1M \sum_{j=1}^M f_{\rm circle}(t;p,q,j),  
$$
 where $f_j^{\rm one-sided}(t;p,q,M)$ is the adoption probability of node~$j$ in a one-sided line with $M$ nodes, and $f_{\rm circle}(t;p,q,j)$ is given by~\eqref{1D_Left_Explicit-periodic} with $M=j$.  
\end{lemma}
\begin{proof}
By the indifference principle, the probability that node~$j$ did not adopt by time~$t$ is the same in the two networks shown in Figure~\ref{fig:indifference_line}. By Figure~\ref{fig:indifference_line}B,   
\begin{equation*}
 \text{Prob}(X_j(t) = 0) = S^{\rm one-sided}(t;p,q,M=j),
\end{equation*}
or equivalently
$$
\text{Prob}(X_j(t) = 1) = 1-S^{\text{\rm one-sided}}(t;p,q,M=j) = f_{\rm circle}(t;p,q,j).
$$
Hence, the result follows from~\eqref{eq:f_sum}.
\end{proof}
\begin{figure}[ht!]
\begin{center}
\scalebox{0.7}{\includegraphics{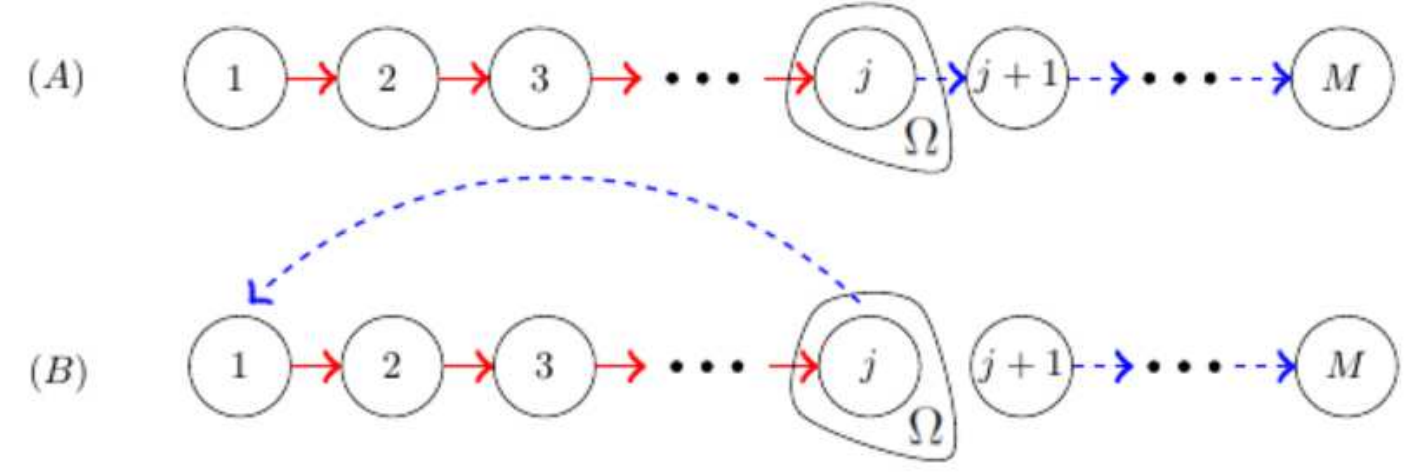}}
\caption{Equivalent networks for the calculation of $\emph{\text{Prob}}(X_j(t) = 0)$ on the one-sided line,  see proof of Lemma~\emph{\ref{lem:f_one-sided_line}}. Solid red and dashed blue arrows are influential and 
non-influential edges to~$\Omega = \{j\}$, respectively. \emph{(}A\emph{)} One-sided line. \emph{(}B\emph{)} The non-influential edge \protect\circled{$j$} $\to$ \protect\circled{$j+1$} is ``replaced'' with the non-influential edge 
\protect\circled{$j$} $\to$ \protect\circled{$1$}. 
}
\label{fig:indifference_line}
\end{center}
\end{figure}

One-sided diffusion on a line is slower than on a circle. This difference, however, disappears as $M\to \infty$:
\begin{lemma}
For $t,p,q,M>0$, 
\begin{equation}
  \label{eq:f_one-sided_line<circle}
f^{\rm one-sided}_{\rm line}(t;p,q,M) <
      f_{\rm circle}(t;p,q,M), 
\end{equation}
but
\begin{equation}
  \label{eq:f_one-sided_line-M=infinity}
\lim_{M \to \infty} f^{\rm one-sided}_{\rm line}(t;p,q,M) =
      \lim_{M \to \infty} f_{\rm circle}(t;p,q,M) = f_{\rm 1D}(t).  
\end{equation}
\end{lemma}
\begin{proof}
Relation~\eqref{eq:f_one-sided_line<circle} is a consequence of the dominance principle, see Corollary~\ref{cor:f_a<f_B}.
We note that if $ a_\infty:=\lim_{n \to \infty} a_n$ then $\lim_{n \to \infty} \frac1n \sum_{k=1}^n a_k = a_\infty$.
Therefore, relation~\eqref{eq:f_one-sided_line-M=infinity} follows from Lemmas~\ref{lem:1D_Left_Explicit_prop-M=infinity} and~\ref{lem:f_one-sided_line}.
\end{proof}

Finally, we consider the two-sided line case.
\begin{lemma}
\label{lem:f_two-sided_line}
The expected fraction of adopters on a two-sided line with $M$ nodes, see~\eqref{eq:GeneralAdoptionRates_Left_Right-line}, is given by
\begin{subequations}
\begin{flalign}
  \label{eq:f_two-sided_line}
f_{\rm line}^{\rm two-sided}(t;p,q,M):= \frac{1}{M}\sum_{j=1}^{M} f_j^{\rm two-sided}(t;p,q,M),
\end{flalign}
where 
\begin{flalign}
f_j^{\rm two-sided}(t;p,q,M) = 
\begin{cases}
f_{\rm circle}\left(t;p,\frac{q}{2},M\right),\quad &\text{$j=1, M$} \\
1-e^{-(p+q)t}\left(1+\frac{q}{2}A_j(t)\right), \quad &\text{$2 \le j \le M-1$}
\end{cases}
\end{flalign}
is the adoption probability of node $j$ in a two-sided line with $M$ nodes, $f_{\rm circle}$ is given by~\eqref{1D_Left_Explicit-periodic}, 
\begin{equation}
\label{eq:two_sided_line_s_tilde}
\begin{aligned}
A_j\left(t\right)& = \int_{0}^{t}e^{(p+q)\tau}\Big[S\left(\tau;p,\frac{q}{2},j\right)S\left(\tau;p,\frac{q}{2},M-j\right)\\
&\qquad \quad ~~+ S\left(\tau;p,\frac{q}{2},j-1\right)S\left(\tau;p,\frac{q}{2},M-j+1\right)\Big]d\tau .
\end{aligned}
\end{equation}
\end{subequations}
and $S=1-f_{\text{circle}}$, see \eqref{eq:s=1-f}.
\end{lemma}
\begin{proof}
We first consider the boundary nodes $j=1,M$. By the indifference principle, we can calculate the probability
that the right boundary node did not adopt by time~$t$ using the equivalent network in Figure~\ref{fig:indifference_two_sided_line}B.
Therefore, $\text{Prob}\left(X_{M}(t)=0\right)=S\left(t;p,\frac{q}{2},M\right)$.
By symmetry,  $\text{Prob}\left(X_{M}(t)=0\right)=\text{Prob}\left(X_{1}(t)=0\right)$. 
Therefore, 
\begin{equation}
\label{eq:two_sided_line_x1}
\text{Prob}\left(X_{1}(t)=1\right)=\text{Prob}\left(X_{M}(t)=1\right)=f_{\rm circle}\left(t;p,\frac{q}{2},M\right).
\end{equation}

\begin{figure}[ht!]
\begin{center}
\scalebox{0.65}{\includegraphics{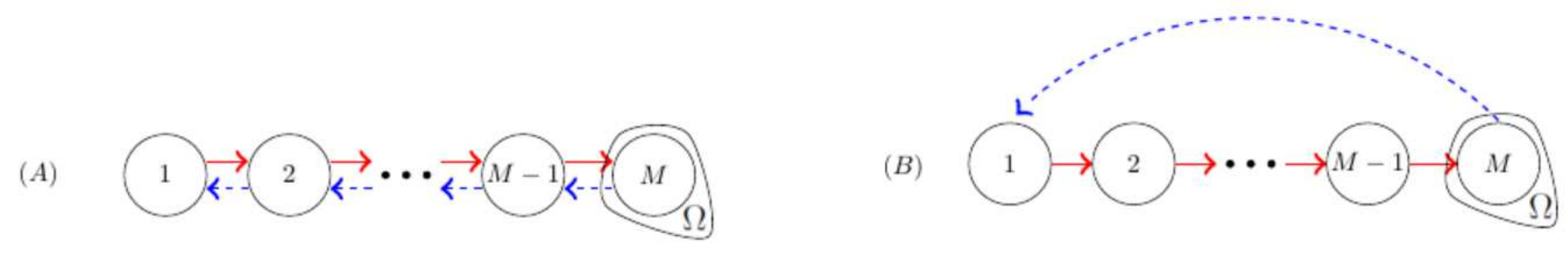}}
\caption{Equivalent networks for the calculation of $\emph{\text{Prob}}\left(X_{M}(t)=0\right)$ on the two-sided line. Solid red and dashed blue arrows are influential and 
non-influential edges to~$\Omega = \{M\}$, respectively. \emph{(}A\emph{)} Two-sided line. \emph{(}B\emph{)} The $M-1$ non-influential edges 
\protect\circled{$M$} $\to$ \protect\circled{$M-1$} $\to \cdots \to $ \protect\circled{$1$}
which ``point away'' from $\Omega$ are  deleted. The non-influential edge \protect\circled{$M$} $\to$ \protect\circled{$1$} is added. }

\label{fig:indifference_two_sided_line}
\end{center}
\end{figure}

Next, we consider the interior nodes $j=2, \dots, M-1$. The evolution equation for $\text{Prob}\left(X_j(t)=0\right)$ is, see Appendix \ref{app:two_sided_int_node},
\begin{equation}
\label{eq:two_sided_evolution}
\begin{aligned}
\frac{\text{d}}{\text{dt}}&\text{Prob}\left(X_j(t)=0\right)=-\left(p+q\right)\text{Prob}\left(X_j(t)=0\right)\\
&+\frac{q}{2}\big[\text{Prob}\left(X_{j-1}(t)=0,X_{j}(t)=0\right)+\text{Prob}\left(X_{j}(t)=0,X_{j+1}(t)=0\right)\big]. 
\end{aligned}
\end{equation}
By the indifference principle, we can calculate $\text{Prob}\left(X_{j-1}(t)=0,X_{j}(t)=0\right)$ 
 from Figure~\ref{fig:indifference_two_sided_line_middle_node}B. In that network, the states of $j-1$ and $j$ are independent, $j-1$ belongs to a one-sided circle with $j-1$ nodes, and $j$ belongs to a one-sided circle with $M-j+1$ nodes.
 Therefore,
\begin{equation}
\label{eq:two_sided_line_j_j+1}
\text{Prob}\left(X_{j-1}(t)=0,X_{j}(t)=0\right) = S\left(t;p,\frac{q}{2},j-1\right)S\left(t;p,\frac{q}{2},M-j+1\right) 
\end{equation}
for $j=2,\ldots,M$. Plugging   
this into~\eqref{eq:two_sided_evolution} and solving the ODE for $\text{Prob}\left(X_j(t)=0\right)$ yields
\begin{equation}
\label{eq:two_sided_line_xj}
\text{Prob}\left(X_j(t)=0\right) = e^{-(p+q)t}+\frac{q}{2}e^{-(p+q)t}A_j(t), \qquad 2 \leq j\leq M-1,
\end{equation}
where $A_j(t)$ is defined in \eqref{eq:two_sided_line_s_tilde}. The desired result follows from~\eqref{eq:f_sum}, \eqref{eq:two_sided_line_x1}, and~\eqref{eq:two_sided_line_xj}. 
%
\end{proof}

\begin{figure}[ht!]
\begin{center}
\scalebox{0.75}{\includegraphics{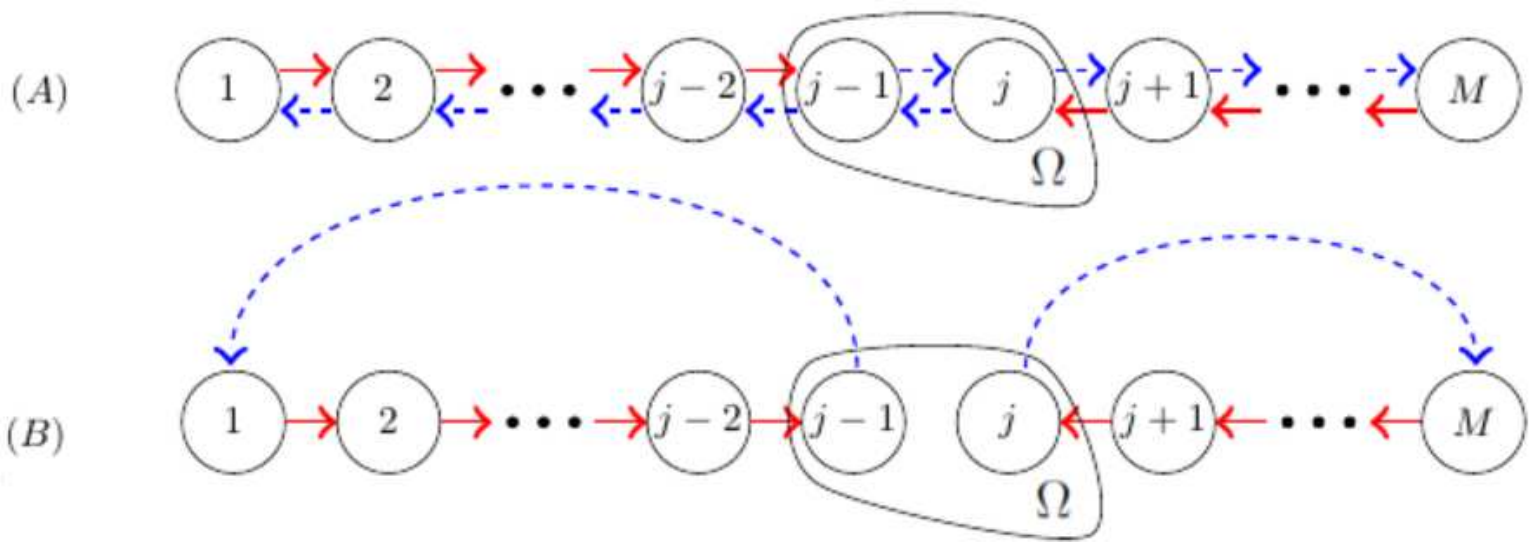}}
\caption{Equivalent networks for the calculation of $\emph{\text{Prob}}\left(X_{j-1}(t)=0,X_{j}(t)=0\right)$. Solid red and dashed blue arrows are influential and 
non-influential edges to $\Omega = \left\{j-1,j\right\}$, respectively. \emph{(}A\emph{)} Two-sided line. \emph{(}B\emph{)} All non-influential edges are deleted. The two non-influential edges 
\protect\circled{$j-1$} $\to$ \protect\circled{$1$} and \protect\circled{$j$} $\to$ \protect\circled{$M$} are added. 
\label{fig:indifference_two_sided_line_middle_node}
}
\end{center}
\end{figure}

Two-sided diffusion is (also) slower on a line than on a circle:
\begin{lemma}
For $t,p,q>0$, 
\begin{equation}
  \label{eq:f_two-sided_line<circle}
f^{\rm two-sided}_{\rm line}(t;p,q,M) <
      f_{\rm circle}(t;p,q,M). 
\end{equation}
\end{lemma}
\begin{proof}
This is a consequence of the dominance principle, see Corollary~\ref{cor:f_a<f_B}.
\end{proof}

\subsection{Hybrid network (circle with a ray)}
\label{sec:hybrid}

We can use the indifference principle to compute the adoption curve on hybrid networks.
For example, consider a one-sided circle with $M-K$~nodes, from which issues a one-sided ray with $K$~nodes (Figure~\ref{fig:circle_with_line2}).
All nodes and edges have the external and internal parameters of~$p$ and~$q$, respectively. 

\begin{figure}[ht!]
\begin{center}
\scalebox{0.7}{\includegraphics{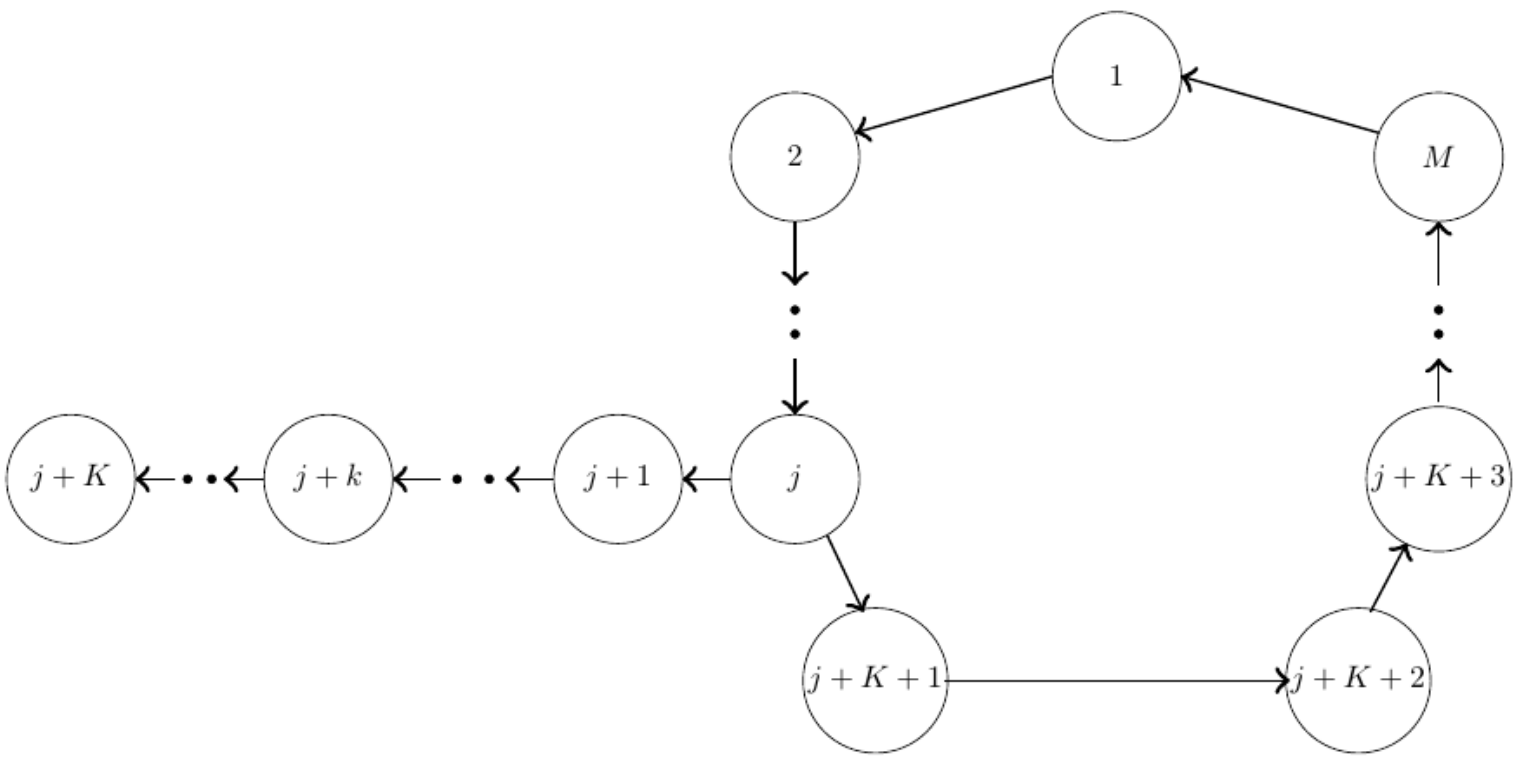}}
\caption{A one-sided circle with $M-K$~nodes, from which issues a one-sided ray with $K$~nodes.}
\label{fig:circle_with_line2}
\end{center}
\end{figure}
\begin{lemma}
\label{lem:circle_with_line}
The expected fraction of adopters in the hybrid circle-line network is
$$f_{\emph{\text{hybrid}}}(t) = \frac{1}{M}\left[\left(M-K\right)f_{\emph{\text{circle}}}(t;p,q,M-K)+\sum_{k=1}^{K}f_{\emph{\text{circle}}}(t;p,q,M-K+k)\right],
$$
where $f_{\emph{\text{circle}}}$ is given by \eqref{1D_Left_Explicit-periodic}.

\end{lemma}
\begin{proof}
{
See Appendix \ref{app:hybrid}.
}
\end{proof}

\subsection{$f_{\rm line}^{\rm one-sided}<f_{\rm line}^{\rm two-sided}$} 
\label{sec:f_one-sided<f_two-sided}

 In Lemma~\ref{lem:circle_two-sided} we proved that the one-sided and two-sided diffusion on the circle coincide. This equivalence was 
confirmed numerically in Figure~\ref{fig:cellular_1D_one_versus_two_sided_main}A.
Repeating this simulation on the line, however, suggests that one-sided diffusion is strictly slower that two-sided diffusion
(Figure~\ref{fig:cellular_1D_one_versus_two_sided_main}B).  
 To analytically prove this result, it suffices to show that, see~\eqref{eq:f_sum},  
\begin{equation}
\label{eq:prob_two<one}
\sum_{j=1}^{M}\text{Prob}\left(X^{\text{two-sided}}_{j}(t)=0\right)<\sum_{j=1}^{M}\text{Prob}\left(X^{\text{one-sided}}_{j}(t)=0\right), \qquad t>0.
\end{equation}
Obviously, this result would hold if $\text{Prob}\left(X^{\text{two-sided}}_{j}(t)=0\right)<\text{Prob}\left(X^{\text{one-sided}}_{j}(t)=0\right)$ for $j=1, \dots. M$.
This, however, is not the case, as is confirmed numerically in e.g.\ Figure~\ref{fig:one_sided_vs_two_sided}A, and analytically in the following lemma:  
\begin{lemma}
\label{lem:X1XM} 
  $
\emph{\text{Prob}}\left(X^{\rm two-sided}_{1}(t)=0\right)<\emph{\text{Prob}}\left(X^{\rm one-sided}_{1}(t)=0\right)$, but
$\emph{\text{Prob}}\left(X^{\rm two-sided}_{M}(t)=0\right)>\emph{\text{Prob}}\left(X^{\rm one-sided}_{M}(t)=0\right).$
\end{lemma}
\begin{proof}
See Appendix \ref{app:X1XM}, which makes use of the generalized dominance and the indifference principles.
\end{proof}
\begin{figure}[ht!]
\begin{center}
\scalebox{0.4}{\includegraphics{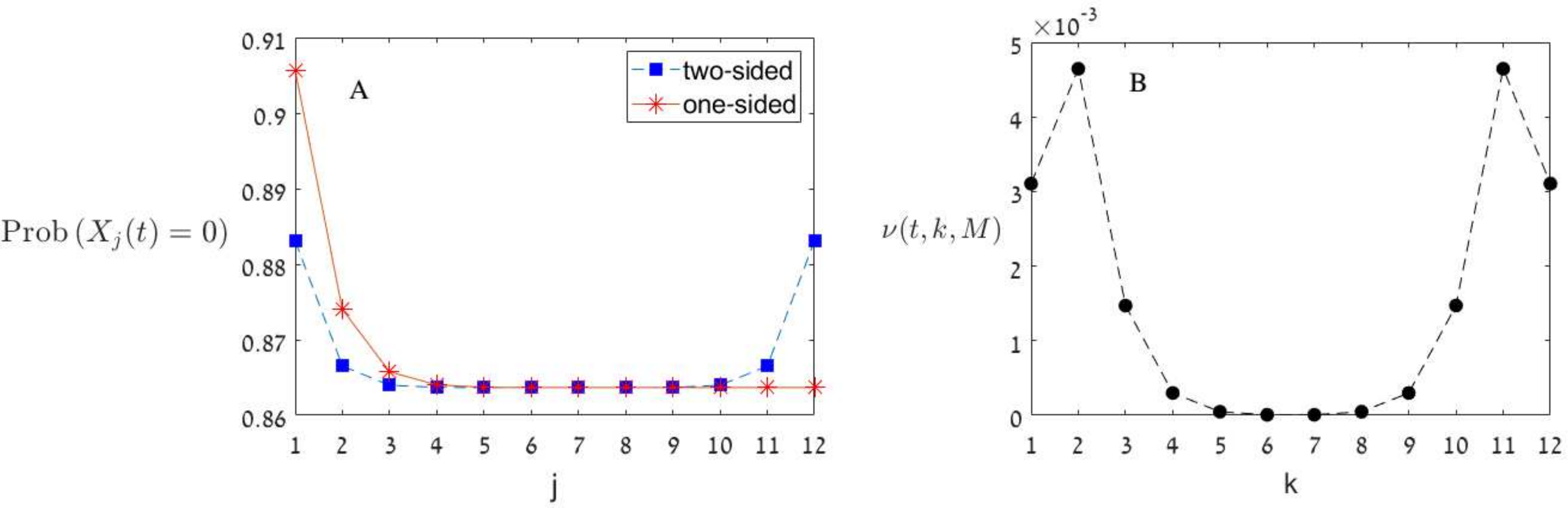}}
\caption{{ \emph{(}A\emph{)}~Probability to remain a non-adopter by time~$t$, on a one-sided (red asterisks) and a two-sided (blue squares) line with~$M$ nodes. \emph{(}B\emph{)}~$\nu(t,k,M)$ as a function of~$k$. In both figures~$M=12$, $t=10$, $p=0.01$ and~$q=0.1$.}}
\label{fig:one_sided_vs_two_sided}
\end{center}
\end{figure}

The key to proving~\eqref{eq:prob_two<one} is to show that for any {\em  pair of symmetric nodes} $\{k,M+1-k\}$, the sum of their  adoption probabilities 
in the one-sided case is smaller than in the two-sided case, i.e.,  
\begin{subequations}
\begin{equation}
\label{eq:nu(t,k)>0}
\nu(t,k,M)>0, \qquad t>0, \quad k = 1, \dots, M,
\end{equation}
 where 
\begin{equation}
\label{eq:nu_def}
\begin{aligned}
\nu(t,k,M) :=& \Big[\text{Prob}(X_{k}^{\rm one-sided}(t)=0)+\text{Prob}(X_{M-k+1}^{\rm one-sided}(t)=0)\Big]\\
&-\Big[\text{Prob}(X_{k}^{\rm two-sided}(t)=0)+\text{Prob}(X_{M-k+1}^{\rm two-sided}(t)=0)\Big], 
\end{aligned}
\end{equation}
\end{subequations}
see e.g.\ Figure~\ref{fig:one_sided_vs_two_sided}B.

We first prove~\eqref{eq:nu(t,k)>0} for the boundary nodes $\{1,M\}$:

\begin{lemma}
\label{lem:end_nodes2}
Let $M\geq 2$.  Then $\nu(t,1,M)>0$ for $t>0$. 
\end{lemma}
\begin{proof}
See Appendix \ref{app:end_nodes2}.
\end{proof}

By symmetry, see~\eqref{eq:nu_def},
$
\nu(t,k,M)=\nu(t,M+1-k,M).
$
Therefore, we only need to prove that $\nu(t,k,M)>0$ for $2 \leq k\leq \frac{M+1}{2}$. 

We first provide an intuitive induction-type argument for~\eqref{eq:nu(t,k)>0}, namely that if $\nu(t,k-1,M)>0$, then $\nu(t,k,M+2)>0$.
 Consider the symmetric pair $\{k,M+2+1-k\}$ in a line with $M+2$ nodes. If we ignore the 
influence of the boundary nodes $\{1,M+2\}$, then the adoption probabilities of nodes $\{k,M+2+1-k\}$ in a line with $M+2$ nodes are given by the adoption probabilities of the symmetric pair $\{k-1,M+1-(k-1)\}$ in a line with $M$ nodes. Therefore, by the induction assumption, $\nu(t,k,M+2)|_{\{X_1(t)= X_{M+1}(t) \equiv 0\}}   = \nu(t,k-1,M)>0$.
 To add the influence of the boundary nodes $\{1,M+2\}$, we should only consider the case where their adoptions are external, i.e.,  
not influenced by the nodes $\{2,M+1\}$. Since the external adoptions of nodes $\{1,M+2\}$  are identical in both networks, and since the combined influence of the nodes $\{1,M+2\}$ on the nodes $\{2,M+1\}$ is the same in both networks, adding their effect does not change the result  that $\nu(t,k,M+2)>0$.

The rigorous proof of~\eqref{eq:nu(t,k)>0} is provided by  
\begin{lemma}
\label{lem:end_nodes}
Let $M\geq 3$. Then
$\nu(t,k,M)>0$ for $t>0$ and $2\leq k\leq \frac{M+1}{2}$.
\end{lemma}
\begin{proof}
See Appendix \ref{app:end_nodes}.
\end{proof}

Lemmas \ref{lem:end_nodes2} and \ref{lem:end_nodes} immediately lead to

\begin{theorem}
 \label{lem:f_one-sided<f_two-sided}
Consider a line with $M\geq 2$ consumers. Then
\begin{equation}
  \label{eq:f_line_1sided<2_sided}
 f_{\rm line}^{\rm one-sided}(t;p,q,M) < f_{\rm line}^{\rm two-sided}(t;p,q,M), \qquad t,p,q>0.
\end{equation}
In addition, 
\begin{equation}
  \label{eq:f_line_1sided<2_sided_M=infty}
\lim_{M \to \infty} f_{\rm line}^{\rm one-sided}(t;p,q,M) =\lim_{M \to \infty} f_{\rm line}^{\rm two-sided}(t;p,q,M).
\end{equation}
\end{theorem}

\begin{proof}
{Lemmas \ref{lem:end_nodes2} and \ref{lem:end_nodes}} imply~\eqref{eq:prob_two<one}.
Combining~\eqref{eq:f_sum} with~\eqref{eq:prob_two<one} yields \eqref{eq:f_line_1sided<2_sided}.

Equation \eqref{eq:f_line_1sided<2_sided_M=infty} is obtained as follows. By~\eqref{eq:f_two-sided_line<circle} and~\eqref{eq:f_line_1sided<2_sided},
$$
 f_{\rm line}^{\rm one-sided}(t;p,q,M) < f_{\rm line}^{\rm two-sided}(t;p,q,M)< f_{\rm circle}(t;p,q,M).
$$  
Taking the limit $M \to \infty$ and using~\eqref{eq:f_2sided_circle=f_1sided_circle_M2infty} and~\eqref{eq:f_one-sided_line-M=infinity}  gives~\eqref{eq:f_line_1sided<2_sided_M=infty}.
\end{proof}


Therefore, {\em the equivalence of one-sided and two-sided diffusion requires that the network be periodic}.
The difference between one-sided and two-sided diffusion initially increases with time, 
see Figure~\ref{fig:cellular_1D_one_versus_two_sided_main}B, as
the probability that adopters reach the boundary increases. Since $\lim_{t \to  \infty} f = 1$
in both cases, however, this difference vanishes as $t \to \infty$.

\subsection{$D \ge 2$}
 \label{sec:D>=2}

  In higher dimensions, the analysis becomes much harder, and so we resort to numerics. 
In Figure~\ref{fig:cellular_2D_one_versus_two_sided_main} we simulate the diffusion on a two-dimensional Cartesian grid.
In the periodic case (a two-dimensional torus), one-sided diffusion and two-sided diffusion are identical
(Figure~\ref{fig:cellular_2D_one_versus_two_sided_main}A).  
In the non-periodic case (a two-dimensional square), however, one-sided diffusion is strictly slower than two-sided diffusion
(Figure~\ref{fig:cellular_2D_one_versus_two_sided_main}B). 
In Figure~\ref{fig:cellular_3D_one_versus_two_sided_main} we observe similar results for diffusion on a three-dimensional Cartesian grid.
Therefore, based on Lemmas~\ref{lem:circle_two-sided} and Theorem~\ref{lem:f_one-sided<f_two-sided}, and Figures~\ref{fig:cellular_2D_one_versus_two_sided_main} and~\ref{fig:cellular_3D_one_versus_two_sided_main},
we formulate 
\begin{conj}
   In the discrete Bass model on a $D$-dimensional Cartesian network:
\begin{enumerate}
  \item One-sided and two-sided diffusion are identical when the network is periodic.
  \item  One-sided diffusion is strictly slower than two-sided diffusion when the network is non-periodic.
 \end{enumerate}  
\end{conj}

\begin{figure}[ht!]
\begin{center}
\scalebox{0.7}{\includegraphics{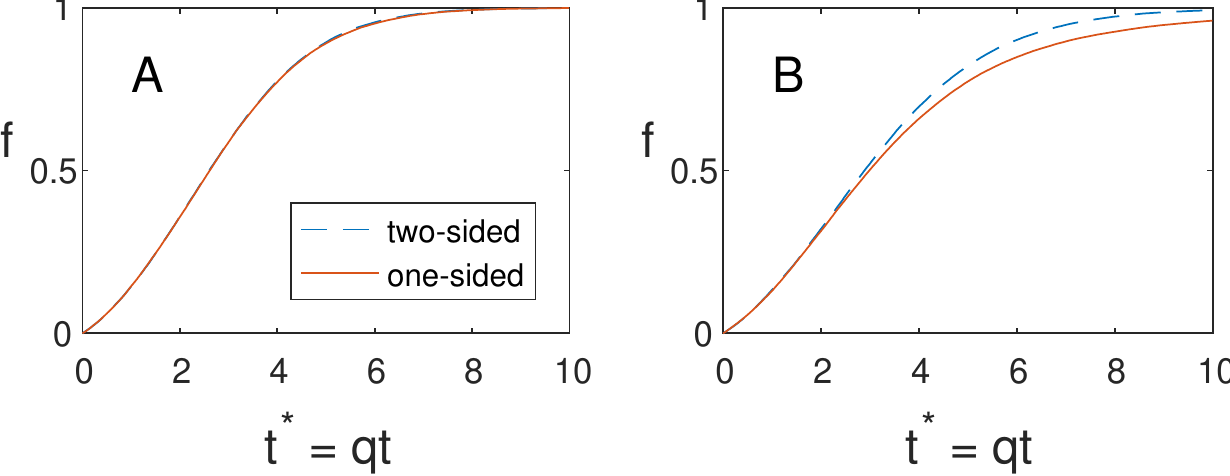}}
\caption{Same as \emph{Figure~\ref{fig:cellular_1D_one_versus_two_sided_main}} on a two-dimensional Cartesian grid with $M = 6 \times 6$ nodes. \emph{(}A\emph{)}~Periodic boundary conditions. \emph{(}B\emph{)}~Non-periodic boundary conditions.}
\label{fig:cellular_2D_one_versus_two_sided_main}
\end{center}
\end{figure}

\begin{figure}[ht!]
\begin{center}
\scalebox{0.7}{\includegraphics{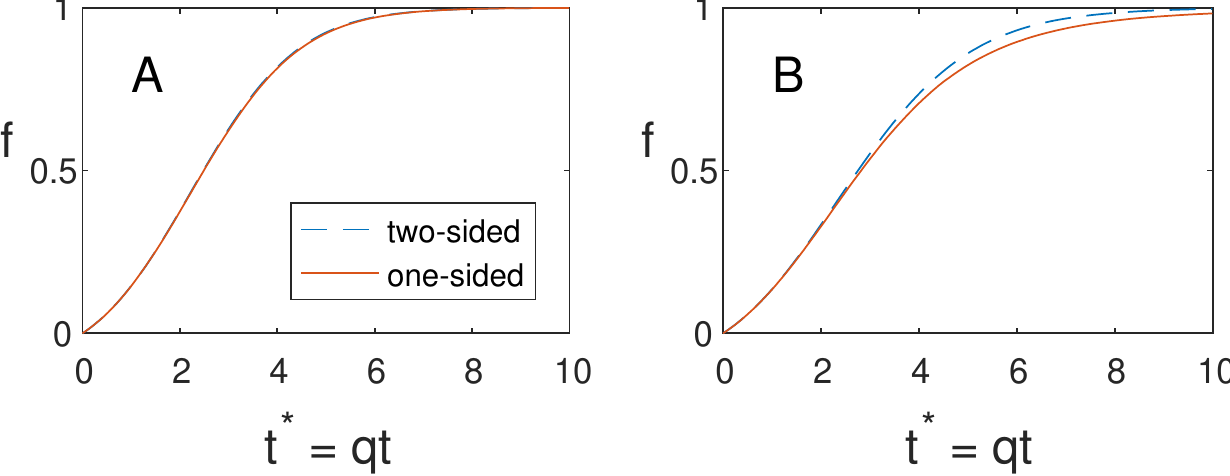}}
\caption{Same as \emph{Figure~\ref{fig:cellular_1D_one_versus_two_sided_main}} on a three-dimensional Cartesian grid with $M = 6 \times 6 \times 6$ nodes. \emph{(}A\emph{)}~Periodic boundary conditions. \emph{(}B\emph{)}~Non-periodic boundary conditions.}
\label{fig:cellular_3D_one_versus_two_sided_main}
\end{center}
\end{figure}

\appendix

\section{End of proof of Lemma~\ref{lem:1D_Left_Explicit_prop}}
\label{app:1D_Left_Explicit_prop}

Substituting the expression for $S(t;M-1)$ from~\eqref{eq:SM} 
into the right-hand side of~\eqref{final_form}, integrating, and equating the coefficients of the exponents on both sides of~\eqref{final_form}  
gives, after some algebra,
\begin{equation}
 \label{eq:SM-RHS}
\begin{aligned}
 S(t;M)  & = e^{-(p+q)t}+qe^{-(p+q)t}\sum_{k=1}^{M-2} A_{k,M-1}\int_0^te^{-kp\tau}d\tau+qe^{-(p+q)t}B_{M-1}\int_0^te^{-(M-1)p\tau+q\tau}d\tau\\
& = e^{-(p+q)t}-qe^{-(p+q)t}\sum_{k=1}^{M-2} A_{k,M-1}\frac{e^{-kpt}-1}{kp}+qe^{-(p+q)t}B_{M-1}\frac{e^{-(M-1)pt+qt}-1}{q-(M-1)p} .
\end{aligned}
\end{equation}
Equating the coefficients of $e^{-Mpt}$ in~\eqref{eq:SM} and~\eqref{eq:SM-RHS} gives
$
B_M = \frac{q}{q-(M-1)p}B_{M-1}.
$  
Since $B_1 = 1$, see~\eqref{eq:S_M=1},  we get that 
\begin{equation}
  \label{eq:BM}
B_M = \frac{q^{M-1}}{\prod_{j=1}^{M-1}(q-jp)}.
\end{equation} 
Equating the coefficients of $e^{-(kp+q)t}$ in~\eqref{eq:SM} and~\eqref{eq:SM-RHS} gives
\begin{equation}
\label{A_km}
A_{k,M} = -\frac{qA_{k-1,M-1}}{(k-1)p}, \qquad k=2,\dots,M-1, 
\end{equation}
and
\begin{equation}
\label{A_1M}
\begin{aligned}
A_{1,M} & = 1+\frac{q}{p}\sum_{k=1}^{M-2}\frac{A_{k,M-1}}{k}-\frac{qB_{M-1}}{q-(M-1)p}
=1+\frac{q}{p}\sum_{k=1}^{M-2}\frac{A_{k,M-1}}{k}-\frac{q^{M-1}}{\prod_{j=1}^{M-1}(q-jp)},
\end{aligned}
\end{equation}
where in the second equality we used~\eqref{eq:BM}.
By~\eqref{A_km} and~\eqref{A_1M}, 
\begin{equation}
\label{A_m-k}
\begin{aligned}
A_{M-k,M} & = \frac{(-q)^{M-k-1}}{(M-k-1)!p^{M-k-1}}A_{1,k+1}\\
& = \frac{(-q)^{M-k-1}}{(M-k-1)!p^{M-k-1}}\left[1-\frac{q^k}{\prod_{j=1}^{k}(q-jp)}+\frac{q}{p}\sum_{j=1}^{k-1}\frac{A_{j,k}}{j}\right].
\end{aligned}
\end{equation}
Using~\eqref{A_km} again, we get that
\begin{equation}
\label{A_jk}
A_{j,k} = \frac{p^{M-k}(M-k+j-1)!}{(-q)^{M-k}(j-1)!}A_{M-k+j,M}.
\end{equation}
Plugging \eqref{A_jk} into \eqref{A_m-k} yields
$$
A_{M-k,M} = \frac{(-q)^{M-k-1}}{(M-k-1)!p^{M-k-1}}\left[1-\frac{q^k}{\prod_{j=1}^{k}(q-jp)}+\frac{q}{p}\sum_{j=1}^{k-1}\frac{p^{M-k}(M-k+j-1)!}{(-q)^{M-k}j!}A_{M-k+j,M}\right] .
$$
Therefore, 
\begin{equation}
\label{A_m-k2}
A_{M-k,M}=  \frac{(-q)^{M-k-1}}{(M-k-1)!p^{M-k-1}} c_{M-k},
\end{equation}
where
\begin{equation}
\label{defC}
c_{M-k} := 1-\frac{q^k}{\prod_{j=1}^{k}(q-jp)}-\sum_{j=1}^{k-1}\frac{p^{M-k-1}(M-k+j-1)!}{(-q)^{M-k-1}j!}A_{M-k+j,M}. 
\end{equation}
Relation~\eqref{1D_Left_Explicit-periodic-A} follows from~$f = 1-S$ and from~\eqref{eq:SM}, \eqref{eq:BM}, and~\eqref{A_m-k2} with $(M-k) \longrightarrow k$.

Plugging \eqref{A_m-k2} into \eqref{defC} yields 
\begin{equation*}
c_{M-k} = 1-\frac{q^k}{\prod_{j=1}^{k}(q-jp)}-\sum_{j=1}^{k-1}\frac{p^{-j}(-q)^{j}}{j!}c_{M-k+j}.
\end{equation*}
Making the change of variables $\tilde{j}:=k-j$ in the summation gives~\eqref{c_m-k}.

\section{Proof of~\eqref{eq:two_sided_evolution}} 
\label{app:two_sided_int_node}

Following~\cite[proof of Lemma 8]{OR-10}, we can obtain  from~\eqref{eq:GeneralAdoptionRates_Left_Right-line} that
\begin{equation}
\label{eq:two_sided_line_ode_original}
\begin{aligned}
\frac{\text{d}}{\text{dt}}\text{Prob}\left(X_j(t)=0\right)=&-p\text{Prob}\left(X_{j-1}(t)=0, X_{j}(t)=0, X_{j+1}(t)=0\right)\\
&-\left(p+\frac{q}{2}\right)\text{Prob}\left(X_{j-1}(t)=0, X_{j}(t)=0, X_{j+1}(t)=1\right)\\
&-\left(p+\frac{q}{2}\right)\text{Prob}\left(X_{j-1}(t)=1, X_{j}(t)=0, X_{j+1}(t)=0\right)\\
&-\left(p+q\right)\text{Prob}\left(X_{j-1}(t)=1, X_{j}(t)=0, X_{j+1}(t)=1\right).
\end{aligned}
\end{equation}
The configuration $\{X_{j}(t)=0, X_{j+1}(t)=0\}$ can be written as a union of two disjoint configurations:
\begin{equation*}
\begin{aligned}
\{X_{j}(t)=0,& X_{j+1}(t)=0\} =\\ 
&\{X_{j-1}(t)=0,X_{j}(t)=0, X_{j+1}(t)=0\}\cup \{X_{j-1}(t)=1,X_{j}(t)=0, X_{j+1}(t)=0\}.
\end{aligned}
\end{equation*}
Therefore, it probability is the sum of the probabilities of the disjoint configurations:
\begin{equation}
\label{eq:two_sided_line1}
\begin{aligned}
\text{Prob}\left(X_{j}(t)=0,X_{j+1}(t)=0\right)=&~~~\text{Prob}\left(X_{j-1}(t)=0,X_{j}(t)=0,X_{j+1}(t)=0\right)\\
&+\text{Prob}\left(X_{j-1}(t)=1,X_{j}(t)=0,X_{j+1}(t)=0\right).
\end{aligned}
\end{equation}
Similarly, 
\begin{equation}
\label{eq:two_sided_line2}
\begin{aligned}
\text{Prob}\left(X_{j-1}(t)=0,X_{j}(t)=0\right)=&~~~\text{Prob}\left(X_{j-1}(t)=0,X_{j}(t)=0,X_{j+1}(t)=0\right)\\
&+\text{Prob}\left(X_{j-1}(t)=0,X_{j}(t)=0,X_{j+1}(t)=1\right),
\end{aligned}
\end{equation}
and 
\begin{equation}
\label{eq:two_sided_line3}
\begin{aligned}
\text{Prob}\left(X_{j}(t)=0\right)=&~~~\text{Prob}\left(X_{j-1}(t)=0,X_{j}(t)=0,X_{j+1}(t)=0\right)\\
&+\text{Prob}\left(X_{j-1}(t)=1,X_{j}(t)=0,X_{j+1}(t)=0\right)\\
&+\text{Prob}\left(X_{j-1}(t)=0,X_{j}(t)=0,X_{j+1}(t)=1\right)\\
&+\text{Prob}\left(X_{j-1}(t)=1,X_{j}(t)=0,X_{j+1}(t)=1\right).
\end{aligned}
\end{equation}
Rearranging \eqref{eq:two_sided_line1}, \eqref{eq:two_sided_line2} and \eqref{eq:two_sided_line3}, and plugging the relevant terms of these equations into \eqref{eq:two_sided_line_ode_original}, leads to \eqref{eq:two_sided_evolution}.

A similar derivation yields the relation 
\begin{equation}
\label{eq:evolution_j_j+1_twosided}
\begin{aligned}
\frac{\text{d}}{\text{dt}}\text{Prob}\left(X_j(t)=0,X_{j+1}(t)=0\right) =& -(2p+q)\text{Prob}\left(X_j(t)=0,X_{j+1}(t)=0\right)\\
&+\frac{q}{2}\Big[\text{Prob}\left(X_{j-1}(t)=0,X_{j}(t)=0,X_{j+1}(t)=0\right)\\
&~~ \quad +\text{Prob}\left(X_j(t)=0,X_{j+1}(t)=0,X_{j+2}(t)=0\right)\Big]
\end{aligned}
\end{equation}
which is used in the derivation of~\eqref{eq:psi6}.

\section{Proof of Lemma \ref{lem:circle_with_line}}
\label{app:hybrid}
We first consider node  $j+k$ which is not on the circle, where $k=1, \dots, K$. By the indifference principle, its probability to be a non-adopter by time $t$ can be calculated from the equivalent network in Figure \ref{fig:circle_with_line}(B). Therefore, $\text{Prob}\left(X_{j+k}(t)=0\right)=S(t,p,q,M-K+k)$, or equivalently
\begin{equation}
\label{eq:hybrid_line}
\text{Prob}\left(X_{j+k}(t)=1\right)=f_{\text{circle}}(t;p,q,M-K+k).
\end{equation}
For the $M-K$ nodes on the circle, it can easily be verified that edges \circled{$j$} $\to$ \circled{$j+1$} $\to \ldots$ \circled{$j+K$} are non-influentials to any of them. Therefore, by the indifference principle, the probability of such a node $j$ to become an adopter is
\begin{equation}
\label{eq:hybrid_circle}
\text{Prob}\left(X_{j}(t)=1\right)=f_{\text{circle}}(t;p,q,M-K).
\end{equation}
Combining \eqref{eq:f_sum}, \eqref{eq:hybrid_line}, and \eqref{eq:hybrid_circle} yields the result.

\begin{figure}[ht!]
\begin{center}
\scalebox{0.8}{\includegraphics{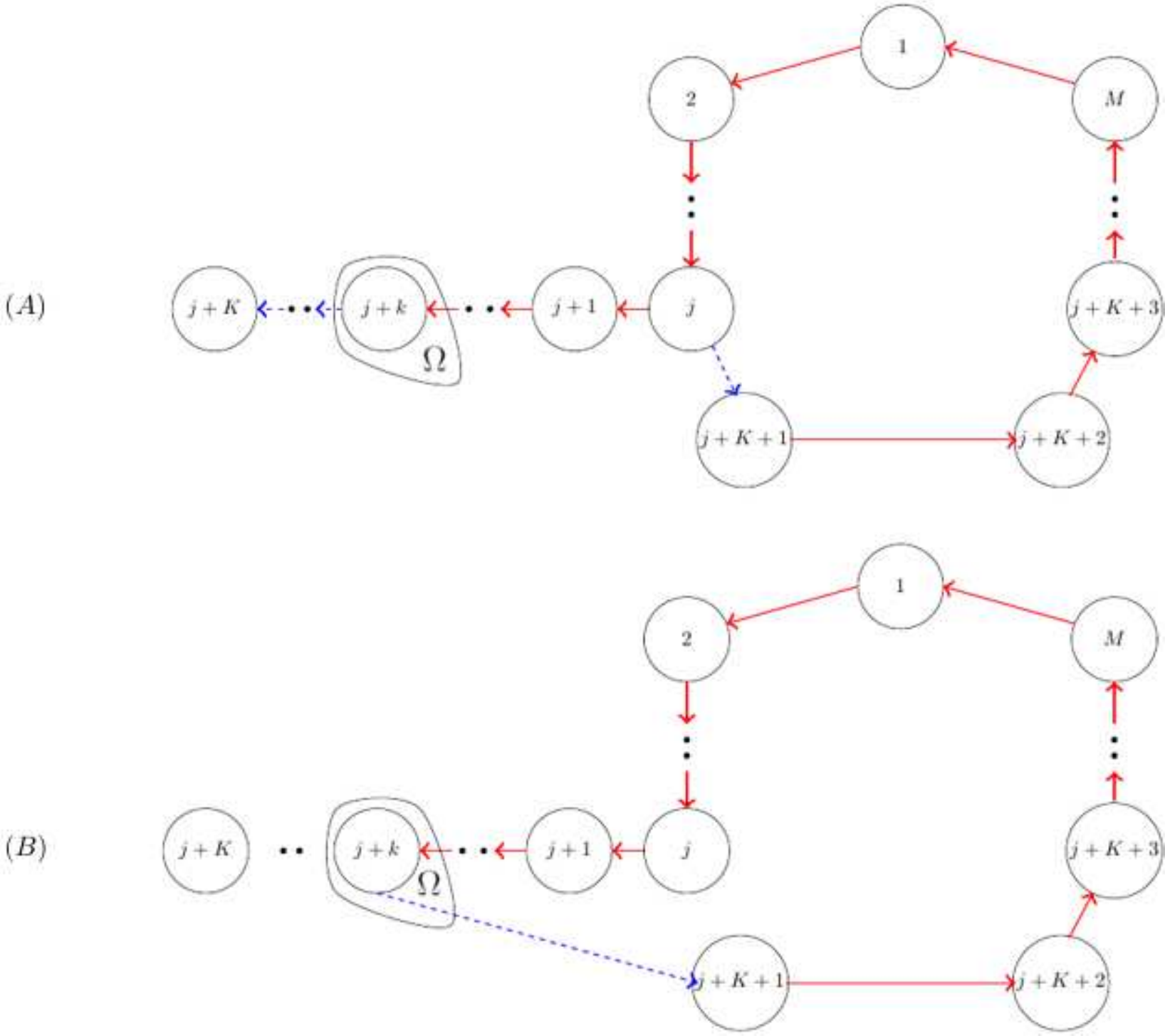}}
\caption{ Equivalent networks for the calculation of $\emph{\text{Prob}}\emph{(}X_{j+k}(t)=0\emph{)}$. Solid red and dashed blue arrows are influential and non-influential edges to $\Omega=\{j+k\}$, respectively. \emph{(}A\emph{)}~One-sided circle with a one-sided ray.  \emph{(}B\emph{)}~The non-influential edges \protect\circled{$j$} $\to$ \protect\circled{$j+K+1$} and \protect\circled{$j+k$} $\to$ \protect\circled{$j+k+1$} are deleted; the non-influential edge \protect\circled{$j+k$} $\to$ \protect\circled{$j+K+1$} is added.}
\label{fig:circle_with_line}
\end{center}
\end{figure}

\section{Proof of Lemma \ref{lem:X1XM}}
\label{app:X1XM}
By the indifference principle for $\Omega = \{1\}$, the networks in Figures  \ref{fig:indifference_line_first_node}(A1) and \ref{fig:indifference_line_first_node}(A2) are equivalent for the one-sided case, and the networks in Figures \ref{fig:indifference_line_first_node}(B1) and \ref{fig:indifference_line_first_node}(B2) are equivalent for the two-sided case. By the strong version of the generalized dominance principle applied to networks \ref{fig:indifference_line_first_node}(A2) and \ref{fig:indifference_line_first_node}(B2), see remark after Lemma \ref{lem:dominance-principle-generalized}, $\text{Prob}\left(X^{\rm two-sided}_1(t)=0\right)<\text{Prob}\left(X^{\rm one-sided}_1(t)=0\right)$. Similarly, by the indifference principle for $\Omega=\{M\}$, the networks in Figures \ref{fig:indifference_line_last_node}(B1) and \ref{fig:indifference_line_last_node}(B2) are equivalent for the two-sided case. By the strong version of the generalized dominance principle applied to networks \ref{fig:indifference_line_last_node}(A) and \ref{fig:indifference_line_last_node}(B2), $\text{Prob}\left(X^{\rm two-sided}_M(t)=0\right)>\text{Prob}\left(X^{\rm one-sided}_M(t)=0\right)$.

\begin{figure}[ht!]
\begin{center}
\vspace{0.5cm}
\scalebox{0.65}{\includegraphics{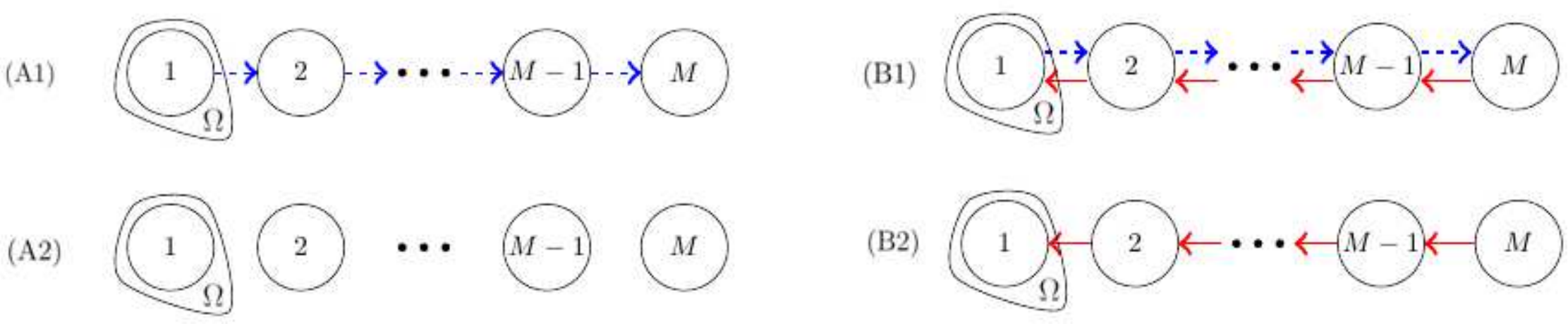}}
\caption{Equivalent networks for the calculation of $\emph{\text{Prob}}\emph{(}X_1(t)=0\emph{)}$ for the one-sided and the two-sided cases.  Solid red and dashed blue arrows correspond to influential and non-influential edges to $\Omega = \{1\}$, respectively. \emph{(}A1\emph{)}~One-sided line. 
\emph{(}A2\emph{)}~All non-influential edges in \emph{(}A1\emph{)} are deleted. \emph{(}B1\emph{)}~Two-sided line. \emph{(}B2\emph{)}~All non-influential edges in \emph{(}B1\emph{)} are deleted.
} 
\label{fig:indifference_line_first_node}
\end{center}
\end{figure}


\begin{figure}[ht!]
\begin{center}
\scalebox{0.65}{\includegraphics{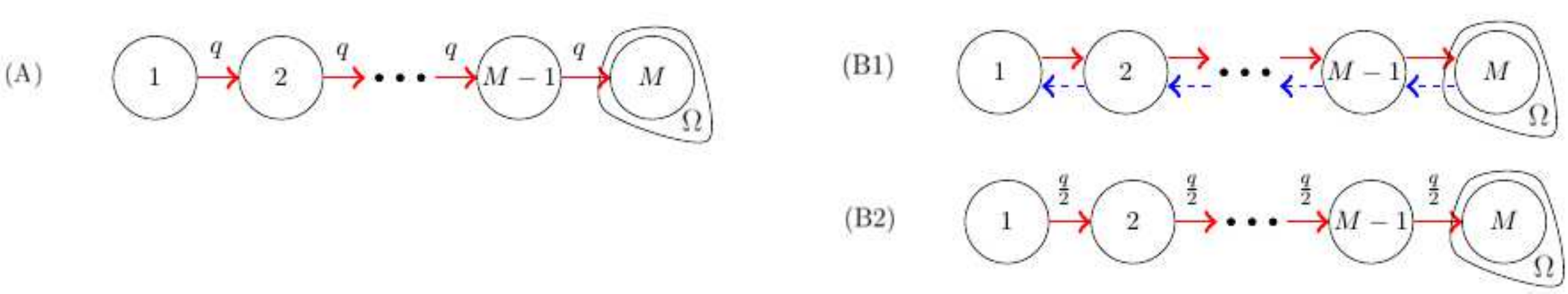}}
\caption{Equivalent networks for the calculation of $\emph{\text{Prob}}\emph{(}X_M(t)=0\emph{)}$ for the one-sided and the two-sided cases.  Solid red and dashed blue arrows correspond to influential and non-influential edges to $\Omega = \{M\}$, respectively. \emph{(}A\emph{)}~One-sided line (with an internal influence of $q$). \emph{(}B1\emph{)}~Two-sided line \emph{(}with an internal influence of $\frac{q}{2}$\emph{)}. \emph{(}B2\emph{)}~All non-influential edges in \emph{(}B1\emph{)} are deleted, resulting in a one-sided line with an internal influence of $\frac{q}{2}$.
} 
\label{fig:indifference_line_last_node}
\end{center}
\end{figure}

\section{Proof of Lemma \ref{lem:end_nodes2}}
\label{app:end_nodes2}
{
We begin with several auxiliary results.
\begin{lemma}
\label{lem:Gronwall}
 Let $\sigma(t)$ be the solution of 
\begin{equation*}
\frac{\text{d}}{\text{dt}}\sigma(t) + K\sigma(t)=b(t), \qquad t>0, \quad \sigma(0)=0,
\end{equation*}
where $K$ is a constant, and $b(t)>0$ for $t>0$. Then
$
\sigma(t)>0 \text{~for~} t>0.
$ 
\end{lemma}
\begin{proof}
Let $t>0$. Since $b(t)>0$,  we have that 
$
\frac{\text{d}}{\text{dt}}\sigma(t)+K\sigma(t)>0.
$
Multiplying both sided by $e^{Kt}$ yields
$\frac{\text{d}}{\text{dt}}\big(e^{Kt}\sigma(t)\big)>0.$
Integrating both sides from $0$ to $t$ and rearranging leads to
$\sigma(t)>e^{-Kt}\sigma(0).
$
Since $\sigma(0)=0$, the result follows.
\end{proof}

\begin{lemma}
\label{lem:gamma}
Let $M\geq 3$, and let
\begin{equation*}
\gamma(t,k,M):=S_k(t;p,q,M)-2S_{k+1}(t;p,q,M)+S_{k+2}(t;p,q,M).
\end{equation*}
Then $\gamma(t,k,M)>0$ for $t>0$ and $k=1, \dots, M-2$.
\end{lemma}
\begin{proof}
By~\cite[Lemma 7]{OR-10},  
\begin{equation*}
\gamma(t,k,M) =
\text{Prob}\left(X_{j+1}(t)=1,X_{j+2}(t)=0,\ldots,X_{j+k+1}(t)=0,X_{j+k+2}(t)=1\right).
\end{equation*}
Since the right-hand side is strictly positive for $t>0$, the result follows.\footnote{{The restriction $M\geq 3$ follows from the term $S_{k+2}(t;p,q,M)$, since $M\geq k+2$  and $k\geq 1$.}}
\end{proof}
\begin{lemma}
\label{lem:beta}
Let $M\geq 2$, and let
\begin{equation*}
\label{eq:beta_def}
\begin{aligned}
\beta(t,k,M):=\big[S_k\big(t;p,\frac{q}{2},M\big)-S_{k+1}\big(t;p,\frac{q}{2},M\big)\big]
 -\big[S_{k}\big(t;p,q,M\big)-S_{k+1}\big(t;p,q,M\big)\big]. 
\end{aligned}
\end{equation*}
Then
\begin{equation}
\label{eq:beta}
\beta(t,k,M)>0, \qquad t>0, \quad k=1, \dots, M-1.
\end{equation}
\end{lemma}
\begin{proof}
We begin by considering the case of $M\geq 3$. We prove \eqref{eq:beta} by a reverse induction on $k$. 
Thus, we first prove that $\beta(t,M-1,M)>0$. Then we show that $\beta(t,k+1,M)>0$ implies $\beta(t,k,M)>0$.

Differentiating $\beta(t,M-1,M)$, and using~\eqref{eq:sk_deriviative} and
 \begin{equation}
\label{eq:sm_solution}
S_M(t;p,q,M) = e^{-Mpt},
\end{equation}
see~\eqref{eq:sM_deriviative}, gives
\begin{equation}
  \label{eq:beta_M-1}
\begin{aligned}
\frac{\text{d}}{\text{dt}}\beta(t,M-1,M)  & +\Big[(M-1)p+\frac{q}{2}\Big]\beta(t,M-1,M) = 
\frac{q}{2}\big[S_{M-1}\big(t;p,q,M\big)-S_M\big(t;p,q,M\big)\big].
\end{aligned}
\end{equation}
By~\cite[Lemma 3]{OR-10}, 
\begin{equation*}
S_{M-1}(t;M) - S_{M}(t;M) = 
\text{Prob}\left(X_1(t)=0,X_2(t)=0,\ldots,X_{M-1}(t)=0,X_{M}(t)=1\right).
\end{equation*}
Since the right-hand side  is strictly positive for $t>0$, so is the right-hand side of~\eqref{eq:beta_M-1}. Since $\beta(0,M-1,M)=0$, Lemma \ref{lem:Gronwall} implies that $\beta(t,M-1,M)>0$ for $t>0$.

For $k<M-1$, assume that $\beta(t,k+1,M)>0$ for $t>0$.
Differentiating $\beta(t,k,M)$ and using \eqref{eq:sk_deriviative} yields
$
\frac{\text{d}}{\text{dt}}\beta(t,k,M)+\left(kp+\frac{q}{2}\right)\beta(t,k,M)=\frac{q}{2}\beta(t,k+1,M)+\frac{q}{2}\gamma(t,k,M)
+p\big[S_{k+1}\big(t;p,\frac{q}{2},M\big)-S_{k+1}\big(t;p,q,M\big)\big].
$
By the induction assumption, $\beta(t,k+1,M)>0$. By Lemma~\ref{lem:gamma}, $\gamma(t,k,M)>0$ for $M\geq 3$.
By Lemma \ref{lem:s2<s1}, the last term on the right-hand side is positive. Hence, since $\beta(0,k,M)=0$, it follows from Lemma~\ref{lem:Gronwall} that $\beta(t,k,M)>0$ for $t>0$ and $M \geq 3$.

When $M=2$, $\beta(t,1,2)=S_1\big(t;p,\frac{q}{2},2\big)-S_1\big(t;p,q,2\big)$,
since $S_M$ is independent of~$q$, see~\eqref{eq:sM_deriviative}. Hence, by Lemma \ref{lem:s2<s1}, $\beta(t,1,2)>0$.\footnote{{The restriction $M\geq 2$ follows from the term $S_{k+1}(t;p,q,M)$, since $M\geq k+1$  and $k\geq 1$.}}
\end{proof}

We now show that $S(t;p,q,M)$ is monotonically decreasing with $M$: 

\begin{lemma}
\label{lem:s_decreasing}
Let
\begin{equation}
\label{eq:alpha_def}
\alpha(t,k):=S_1(t;p,q,k)-S_1(t;p,q,k+1). 
\end{equation}
Then $\alpha(t,k)>0$ for $t>0$ and $k\geq 1$.
\end{lemma}
\begin{proof}
We proceed by induction on~$k$. Thus, we first prove that $\alpha(t,1)>0$, and then show that if $\alpha(t,k-1)>0$, then $\alpha(t,k)>0$.

By \eqref{eq:alpha_def}, 
$
\alpha(t,1) = S_1(t;p,q,1)-S_1(t;p,q,2).
$
Since $S_1(t;p,q,1)$ is the probability of a single node with no neighbors to be a non-adopter by time $t$, then $S_1(t;p,q,1)=S_1(t;p,0,2)$.
By Lemma \ref{lem:s2<s1}, $S_1(t;p,0,2)>S_1(t;p,q,2)$.  Therefore, $\alpha(t,1)>0$ .


Now, assume that $\alpha(t,k-1)>0$. Differentiating $\alpha(t,k)$ and using~\eqref{eq:sk_deriviative} and~\eqref{indifference_2} gives
$
\frac{d}{dt}\alpha(t,k) +(p+q)\alpha(t,k)=q\left[S_2(t;p,q,k)-S_2(t;p,q,k+1)\right]
 = qe^{-pt}\left[S_1(t;p,q,k-1)-S_1(t;p,q,k)\right]
 = qe^{-pt}\alpha(t,k-1).
$
By the induction assumption, the right-hand side is positive for $t>0$. In addition, $\alpha(0,k)=0$. Therefore, by Lemma \ref{lem:Gronwall},  $\alpha(t,k)>0$ for $t>0$.
\end{proof}

We are now ready to prove Lemma \ref{lem:end_nodes2}:

\begin{proof}[Proof of Lemma \ref{lem:end_nodes2}]


In Lemma \ref{lem:f_one-sided_line} we proved that 
\begin{equation}
\label{eq:one-sidedj}
\text{Prob}\left(X_k^{\rm one-sided}(t)=0\right)=S_1(t;p,q,k), \qquad k=1, \dots, M.
\end{equation} 
In addition, in Lemma \ref{lem:f_two-sided_line} we showed that 
\begin{equation}
\label{eq:two-sided1M}
\text{Prob}(X_1^{\rm two-sided}(t)=0)=\text{Prob}(X_M^{\rm two-sided}(t)=0) = S_1\left(t;p,\frac{q}{2},M\right).
\end{equation}
Plugging expressions \eqref{eq:one-sidedj} and \eqref{eq:two-sided1M} into \eqref{eq:nu_def} for $k=1$ gives 
\begin{equation}
\label{eq:nu_1}
\nu(t,1,M) = S_1(t;p,q,1)+S_1(t;p,q,M)-2S_1\left(t;p,\frac{q}{2},M\right).
\end{equation}
Differentiating $\nu(t,1,M)$  and using \eqref{eq:sk_deriviative} gives
$
\frac{\text{d}}{\text{dt}}\nu(t,1,M) + p\nu(t,1,M)=q\beta(t,1,M).
$
By Lemma \ref{lem:beta}, $\beta(t,1,M)>0$ for $t>0$ and $M\geq 2$. In addition, $\nu(0,1,M)=0$. Therefore, by Lemma \ref{lem:Gronwall}, 
\begin{equation}
\label{eq:nu_1>0}
\nu(t,1,M)>0, \qquad t>0, \quad M\geq 2.
\end{equation}
\end{proof}

\Remark For $M=1$, we have from~\eqref{eq:S_M=1} and~\eqref{eq:nu_1} that 
\begin{equation}
\label{eq:nu_1M=1}
\nu(t,1,1)=0.
\end{equation}
Therefore, combining \eqref{eq:nu_1>0} and \eqref{eq:nu_1M=1} yields
\begin{equation}
\label{eq:nu1_summery}
\nu(t,1,M)\geq 0, \qquad t>0,\quad M\geq 1.
\end{equation}
This inequality will be used in Appendix~\ref{app:psi}.





}
\section{Proof of Lemma \ref{lem:end_nodes}}
\label{app:end_nodes}
%
%
%
By symmetry, 
\begin{equation}
\label{eq:symmetryfornu2}
\text{Prob}(X_k^{\rm two-sided}(t)=0)=\text{Prob}(X_{M-k+1}^{\rm two-sided}(t)=0), \qquad k=1, \dots, M. 
\end{equation}
Plugging expressions \eqref{eq:one-sidedj} and \eqref{eq:symmetryfornu2} into \eqref{eq:nu_def} gives
\begin{equation*}
\label{eq:nu_k2}
\nu(t,k,M) = S_1(t;p,q,k)+S_1(t;p,q,M-k+1)-2\text{Prob}(X_k^{\rm two-sided}(t)=0).
\end{equation*}
Differentiating $\nu(t,k,M)$ and using \eqref{eq:two_sided_evolution} and \eqref{eq:sk_deriviative} yields 
\begin{equation}
\label{eq:evolution_nu22}
\frac{\text{d}}{\text{dt}}\nu(t,k,M) +(p+q)\nu(t,k,M)=q\psi(t,k,M),
\end{equation}
where
\begin{equation}
\label{eq:psi_def2}
\begin{aligned}
\psi&(t,k,M) = ~S_{2}(t;p,q,k)+S_{2}(t;p,q,M-k+1)\\
&-\text{Prob}(X_{k-1}^{\rm two-sided}(t)=0,X_{k}^{\rm two-sided}(t)=0)-\text{Prob}(X_k^{\rm two-sided}(t)=0,X_{k+1}^{\rm two-sided}(t)=0).
\end{aligned}
\end{equation}
In Appendix \ref{app:psi} we show that 
\begin{equation}
  \label{eq:psi(t,k)>02}
\psi(t,k,M)>0, \qquad t>0, \quad k\geq 2, \quad M\geq 2k-1.
\end{equation}
Therefore, $\psi(t,k,M)>0$ for $t>0$ and $2\leq k \leq \frac{M+1}{2}$.
In addition, $\nu(0,k,M)=0$.
Now applying Lemma \ref{lem:Gronwall} to \eqref{eq:evolution_nu22} gives the desired result.

%
%
%
%
%
%
%
\subsection{Proof of~\eqref{eq:psi(t,k)>02}} 
\label{app:psi}



 We proceed by induction on $k$.
Thus, we first show that $\psi(t,2,M)>0$ for $t>0$ and for $M \geq 2\times 2-1=3$. Then we show that if $\psi(t,k-1,M)>0$ for $M\geq 2(k-1)-1=2k-3$, then $\psi(t,k,M)>0$ for $M \geq 2k-1$.

By \eqref{eq:two_sided_line_j_j+1},
\begin{equation}
\label{eq:psi1}
\begin{aligned}
\psi(t,&k,M) = ~S_{2}(t;p,q,k)+S_{2}(t;p,q,M-k+1)\\
&-S\Big(t;p;\frac{q}{2},k-1\Big)S\Big(t;p;\frac{q}{2},M-k+1\Big)-S\Big(t;p;\frac{q}{2},k\Big)S\Big(t;p;\frac{q}{2},M-k\Big).
\end{aligned}
\end{equation}
By  Lemma \ref{lem:S_k^one-sided},
\begin{equation}
\label{eq:psi2}
S_{2}(t;p,q,k) = e^{-pt}S(t;p,q,k-1),
\quad S_{2}(t;p,q,M-k+1) =e^{-pt}S(t;p,q,M-k).
\end{equation}
By Lemma \ref{lem:s_decreasing}, 
\begin{equation}
\label{eq:psi3}
S\Big(t;p,\frac{q}{2},k\Big)< S\Big(t;p,\frac{q}{2},k-1\Big),
~~ 
S\Big(t;p,\frac{q}{2},M-k+1\Big)<S\Big(t;p,\frac{q}{2},M-k\Big).
\end{equation}
By \eqref{eq:sm_solution},
\begin{equation}
\label{eq:psi4}
S\Big(t;p,\frac{q}{2},1\Big)=S\big(t;p,q,1\big)=e^{-pt}.
\end{equation} 
Plugging \eqref{eq:psi2}, \eqref{eq:psi3}, and~\eqref{eq:psi4} into~\eqref{eq:psi1} with $k=2$, and then using \eqref{eq:nu_1} yields 
\begin{equation*}
\label{eq:psi5}
\psi(t,2,M)>e^{-pt}\left[S\big(t;p,q,1\big)+S\big(t;p,q,M-2\big)-2S\big(t;p,\frac{q}{2},M-2\big)\right]=e^{-pt}\nu(t,1,M-2).
\end{equation*}
By \eqref{eq:nu1_summery}, the right-hand side is not negative for $M-2\geq 1$, i.e. $M\geq 3$. Therefore, $\psi(t,2,M)>0$ for $M\geq 3$.

For $k\geq 3$, assume that $\psi(t,k-1,M)>0$ for $t>0$ and for $M\geq 2k-3$.
 Differentiating $\psi(t,k,M)$ in \eqref{eq:psi_def2} and using \eqref{eq:sk_deriviative} and \eqref{eq:evolution_j_j+1_twosided} yields
\begin{equation}
\label{eq:psi6}
\begin{aligned}
\frac{\text{d}}{\text{dt}}\psi(t,k&,M)+(2p+q)\psi(t,k,M)= q\Big[S_3\big(t;p,q,k\big)+S_3\big(t;p,q,M-k+1\big)\\
&-\frac{1}{2}\text{Prob}\left(X^{\rm two-sided}_{k-2}(t)=0,X^{\rm two-sided}_{k-1}(t)=0,X^{\rm two-sided}_{k}(t)=0\right)\\
&-\frac{1}{2}\text{Prob}\left(X^{\rm two-sided}_{k-1}(t)=0,X^{\rm two-sided}_{k}(t)=0,X^{\rm two-sided}_{k+1}(t)=0\right)\\
&-\frac{1}{2}\text{Prob}\left(X^{\rm two-sided}_{k-1}(t)=0,X^{\rm two-sided}_{k}(t)=0,X^{\rm two-sided}_{k+1}(t)=0\right)\\
&-\frac{1}{2}\text{Prob}\left(X^{\rm two-sided}_{k}(t)=0,X^{\rm two-sided}_{k+1}(t)=0,X^{\rm two-sided}_{k+2}(t)=0\right)\Big].
\end{aligned}
\end{equation}
By Lemma \ref{lem:S_k^one-sided}, for $k \ge 3$, 
$$S_3\big(t;p,q,k\big) = e^{-2pt}S\big(t;p,q,k-2\big), \quad 
S_2\big(t;p,q,k-1\big) = e^{-pt}S\big(t;p,q,k-2\big).
$$
Therefore, 
\begin{subequations}
\label{eq:psi7}
\begin{flalign}
S_3\big(t;p,q,k\big) = e^{-pt}S_2\big(t;p,q,k-1\big). 
\end{flalign}
Similarly,
\begin{flalign}
S_3\big(t;p,q,M-k+1\big) = e^{-pt}S_2\big(t;p,q,M-k\big). 
\end{flalign}
\end{subequations}
By the indifference principle, see Figure \ref{fig:indifference_line_3nodes_interior}, 
\begin{equation}
\label{eq:psi8}
\begin{aligned}
&\text{Prob}\left(X^{\rm two-sided}_{k}(t)=0,X^{\rm two-sided}_{k+1}(t)=0,X^{\rm two-sided}_{k+2}(t)=0\right) =\\ & S\big(t;p,q,1\big)S\Big(t;p,\frac{q}{2},k\Big)S\Big(t;p,\frac{q}{2},M-k-1\Big)= e^{-pt}S\Big(t;p,\frac{q}{2},k\Big)S\Big(t;p,\frac{q}{2},M-k-1\Big),
\end{aligned} 
\end{equation}
\begin{figure}[ht!]
\begin{center}
\scalebox{0.85}{\includegraphics{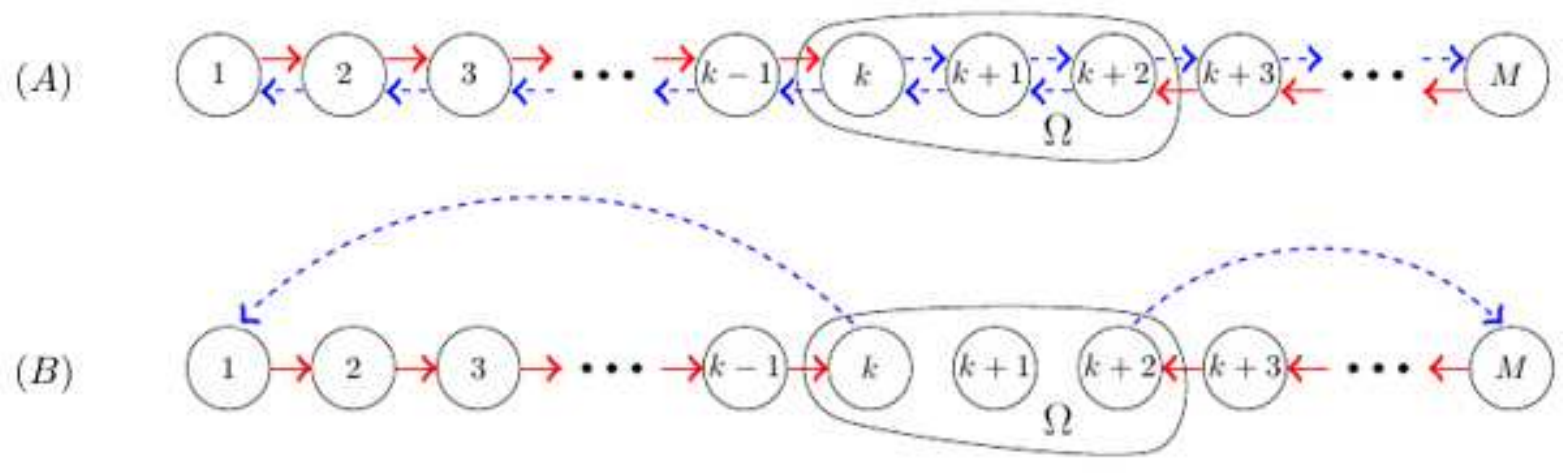}}
\caption{Equivalent networks for the calculation of $\emph{\text{Prob}}(X^{\rm two-sided}_{k}(t)=0,X^{\rm two-sided}_{k+1}(t)=0,X^{\rm two-sided}_{k+2}(t)=0)$ on the two-sided line. Solid red and dashed blue arrows are influential and 
non-influential edges to~$\Omega = \{k,k+1,k+2\}$, respectively. \emph{(}A\emph{)} Two-sided line. \emph{(}B\emph{)} All non-influential edges are  deleted. The non-influential edges \protect\circled{$k$} $\to$ \protect\circled{$1$} and \protect\circled{$k+2$} $\to$ \protect\circled{$M$} are added.}
\label{fig:indifference_line_3nodes_interior}
\end{center}
\end{figure}
where in the second equality we used~\eqref{eq:sm_solution}.
Plugging \eqref{eq:psi7} and \eqref{eq:psi8} into \eqref{eq:psi6} gives
\begin{equation}
\label{eq:psi9}
\begin{aligned}
\frac{\text{d}}{\text{dt}}&\psi(t,k,M)+(2p+q)\psi(t,k,M)=\\
& qe^{-pt}\Big[S_2\big(t;p,q,k-1\big)+S_2\big(t;p,q,M-k\big) -\frac{1}{2}S\big(t;p,\frac{q}{2},k-2\big)S\big(t;p,\frac{q}{2},M-k+1\big)\\
&-\frac{1}{2}S\big(t;p,\frac{q}{2},k-1\big)S\big(t;p,\frac{q}{2},M-k\big)-\frac{1}{2}S\big(t;p,\frac{q}{2},k-1\big)S\big(t;p,\frac{q}{2},M-k\big)\\
&-\frac{1}{2}S\big(t;p,\frac{q}{2},k\big)S\big(t;p,\frac{q}{2},M-k-1\big)\Big].
\end{aligned}
\end{equation}
By Lemma \ref{lem:s_decreasing},
\begin{subequations}
\label{eq:psi10}
\begin{flalign}
S\Big(t;p,\frac{q}{2},k-1\Big)<S\Big(t;p,\frac{q}{2},k-2\Big), \quad S\Big(t;p,\frac{q}{2},k\Big)<S\Big(t;p,\frac{q}{2},k-1\Big),
\end{flalign}
\begin{flalign}
S\Big(t;p,\frac{q}{2},M-k+1\Big)<S\Big(t;p,\frac{q}{2},M-k\Big), \quad S\Big(t;p,\frac{q}{2},M-k\Big)<S\Big(t;p,\frac{q}{2},M-k-1\Big).
\end{flalign}
\end{subequations}
By \eqref{eq:psi9} and \eqref{eq:psi10},
$
\frac{\text{d}}{\text{dt}}\psi(t,k,M)+(2p+q)\psi(t,k,M)> qe^{-pt}\psi(t,k-1,M-2), 
$
where
\begin{equation*}
\begin{aligned}
\psi(t,k-1,M-2) = & ~S_2\big(t;p,q,k-1\big)+S_2\big(t;p,q,M-k\big)
-S\Big(t;p,\frac{q}{2},k-2\Big)S\Big(t;p,\frac{q}{2},M-k\Big)\\
&-S\Big(t;p,\frac{q}{2},k-1\Big)S\Big(t;p,\frac{q}{2},M-k-1\Big),
\end{aligned}
\end{equation*}
see~\eqref{eq:psi1}.
By the induction assumption, the right-hand side is positive for $M-2\geq 2k-3$, i.e. $M\geq 2k-1$. Therefore, by Lemma \ref{lem:Gronwall}, $\psi(t,k,M)>0$ for $M\geq 2k-1$. 


\section*{Acknowledgments}
We thank K. Gillingham for useful discussions. 

\bibliographystyle{siamplain}
\bibliography{diffusion}

\end{document}

%% file: ex_shared.tex
\usepackage{lipsum}
\usepackage{amsfonts}
\usepackage{graphicx}
\usepackage{epstopdf}
\usepackage{algorithmic}
\usepackage{bbm}
\usepackage{amsmath,amssymb}
\usepackage{enumerate}
\ifpdf
  \DeclareGraphicsExtensions{.eps,.pdf,.png,.jpg}
\else
  \DeclareGraphicsExtensions{.eps}
\fi
\usepackage{tikz}
\usetikzlibrary{shapes}
\newcommand*\circled[1]{\tikz[baseline=(char.base)]{
        \node[shape=rounded rectangle,draw,inner sep=2pt] (char) {#1};}}
\def\vomega{\mbox{\boldmath $\omega $}}
\newcommand{\Remark}{\vspace{0mm} \parindent=0pt
         {\bf Remark.} \hspace{0mm} \parindent=3ex}
\newsiamthm{conj}{Conjecture}

\numberwithin{theorem}{section}

\newcommand{\TheTitle}{Boundary Effects in the Discrete Bass Model} 
\newcommand{\TheAuthors}{G. Fibich, T. Levin, and O. Yakir}

\headers{\TheTitle}{\TheAuthors}

\title{{\TheTitle}\thanks{Submitted to the editors 2 January 2018.
\funding{This work was supported by the U.S. Department of Energy under Award DE-EE0007657.}}}

\author{
  Gadi Fibich\thanks{Department of Applied Mathematics, Tel Aviv University, Tel Aviv 6997801, Israel 
(\email{fibich@tau.ac.il},     \email{levintmr@gmail.com},     \email{oren.yakir@gmail.com}).}
  \and
 Tomer Levin\footnotemark[2]
  \and
  Oren Yakir\footnotemark[2]
}

\usepackage{amsopn}
